\documentclass[11pt]{article}

\usepackage[totalwidth=480pt,totalheight=680pt]{geometry}
\usepackage{amsmath,amssymb,amsthm,bbm,graphicx,color,ifthen,hyperref}
\usepackage{natbib}
\bibpunct{(}{)}{;}{a}{,}{,}
\usepackage{ifthen}
\newboolean{preview}
\setboolean{preview}{true}   % Set to true for clean output & references
\ifthenelse{\not\boolean{preview}}{%
  \usepackage{rcs}
  \RCS $Id: RecursiveEM.tex,v 1.69 2011-07-28 08:58:38 cappe Exp $
  \usepackage[notref,notcite]{showkeys}
  \usepackage{fancyhdr}
  \pagestyle{fancyplain}
  
  \lhead{{\tt\footnotesize \RCSId} \hfill \thepage}
  \rhead{}
  \cfoot{}
  
  \newcommand{\note}[2]{\marginpar{\textcolor{blue}{#1} \color{red}\footnotesize{ #2}\color{black}}}%
}
{\newcommand{\note}[2]{}}

\newcommand\rmd{\mathrm{d}}
\newcommand\mcf{\mathcal{F}}
\newcommand{\Yset}{{\mathcal{Y}}}
\newcommand{\nbt}{{d_\theta}} % Dim of parameter vector de parametres
\newcommand{\nbz}{{d_z}} % Dim of regressor
\newcommand{\logl}[2]{\log g(#1;#2)}
\newcommand{\kullback}[2]{\mathrm{K} \left( #1 \left\| #2 \right. \right)}
\newcommand{\lyapunov}{\mathrm{w}}
\newcommand{\meanfield}{\mathrm{h}}

\newcommand{\mixtureweight}[1]{\omega_{#1}}
\newcommand{\mixtureparam}[1]{\lambda_{#1}}
\newcommand{\bmixtureind}[1]{\bar{w}_{#1}}
\newcommand{\hmixtureweight}[1]{\hat{\omega}_{#1}}
\newcommand{\hmixtureparam}[1]{\hat{\lambda}_{#1}}

\newcommand{\FIMcomplete}[3][]%
{\ifthenelse{\equal{#1}{}}{I_{#2}(#3)}{I^{#1}_{#2}(#3)}}
\newcommand{\FIMcond}[2]{J(#1; #2)}
\newcommand{\FIMincomplete}[2][]%
{\ifthenelse{\equal{#1}{}}{I_{\text{obs}}(#2)}{I^{#1}_{\text{obs}}(#2)}}
\newcommand{\funcEM}[3]{Q_{#3}(#1;#2)}
\newcommand{\htheta}{\hat{\theta}}
\newcommand{\functheta}[2][]
{\ifthenelse{\equal{#1}{}}{\bar{\theta}\left(#2\right)}{\bar{\theta}^{#1}\left(#2\right)}}
\def\fx{f}
\def\lx{\log \fx}
\def\fy{g}
\def\dly{{\nabla_{\theta}\log g}}
\def\ls{\ell}

\newcommand\pscal[2] { \left\langle #1 , #2 \right\rangle }
\newcommand{\suffstat}{S}
\newcommand{\condexpsuffstatletter}{\bar{s}}
\newcommand{\condexpsuffstat}[3][]
{\ifthenelse{\equal{#1}{}}{{\condexpsuffstatletter(#2;#3)}}{\left|\condexpsuffstatletter(#2;#3)\right|^{#1}}}
\newcommand{\hcondexpsuffstat}{\hat{s}}
\def\calS{\mathcal{S}}
\def\calK{\mathcal{K}}

\def\rset{\mathbb{R}}

\newcommand{\eqsp}{\;}
\newcommand{\eqdef}{\ensuremath{\stackrel{\mathrm{def}}{=}}}
\newcommand{\rme}{\mathrm{e}}

\def\1{\mathbbm{1}}
\newcommand{\PE}{\mathbb{E}}
\newcommand{\PP}{\mathbb{P}}
\newcommand{\CPE}[3][]
{\ifthenelse{\equal{#1}{}}{\mathbb{E}\left[\left. #2 \, \right| #3 \right]}{\mathbb{E}_{#1}\left[\left. #2 \, \right | #3 \right]}}
\newcommand{\CPP}[3][]
{\ifthenelse{\equal{#1}{}}{\mathbb{P}\left(\left. #2 \, \right| #3 \right)}{\mathbb{P}_{#1}\left(\left. #2 \, \right | #3 \right)}}

\newtheorem{theorem}{Theorem}
\newtheorem{lemma}[theorem]{Lemma}

\newtheorem{proposition}[theorem]{Proposition}

\newtheorem{assumption}[theorem]{Assumption}
\theoremstyle{remark}

\newcommand{\as}{a.s.}
\newcommand{\ie}{i.e.}
\newcommand{\eg}{e.g.}
\title{Online EM Algorithm for Latent Data Models}
\author{Olivier Capp\'e \& Eric Moulines}
\date{LTCI, TELECOM ParisTech, CNRS.\\46 rue Barrault, 75013 Paris, France.}

\begin{document}

\maketitle

\begin{abstract}
  In this contribution, we propose a generic online (also sometimes called adaptive or recursive)
  version of the Expectation-Maximisation (EM) algorithm applicable to latent variable models of
  independent observations. Compared to the algorithm of \cite{titterington:1984}, this approach is
  more directly connected to the usual EM algorithm and does not rely on integration with respect
  to the complete data distribution. The resulting algorithm is usually simpler and is shown to
  achieve convergence to the stationary points of the Kullback-Leibler divergence between the
  marginal distribution of the observation and the model distribution at the optimal rate, \ie,
  that of the maximum likelihood estimator. In addition, the proposed approach is also suitable for
  conditional (or regression) models, as illustrated in the case of the mixture of linear
  regressions model.

\end{abstract}

\note{oKp}{British spelling in use}

%==================================================================
\section{Introduction}
The EM (Expectation-Maximisation) algorithm \citep{dempster:laird:rubin:1977} is a popular tool for
maximum-likelihood (or maximum a posteriori) estimation. The common strand to problems where this
approach is applicable is a notion of {\em incomplete data}, which includes the conventional sense
of missing data but is much broader than that. The EM algorithm demonstrates its strength in
situations where some hypothetical experiments yields \emph{complete} data that are related to the
parameters more conveniently than the measurements are. Problems where the EM algorithm has proven
to be useful include, among many others, mixture of densities
\citep{titterington:smith:makov:1985}, censored data models \citep{tanner:1993}, etc. The EM
algorithm has several appealing properties. Because it relies on complete data computations, it is
generally simple to implement: at each iteration, \emph{(i)} the so-called \emph{E-step} only
involves taking expectation over the conditional distribution of the latent data given the
observations and \emph{(ii)} the \emph{M-step} is analogous to complete data weighted
maximum-likelihood estimation. Moreover, \emph{(iii)} the EM algorithm naturally is an ascent
algorithm, in the sense that it increases the (observed) likelihood at each iteration. Finally
under some mild additional conditions, \emph{(iv)} the EM algorithm may be shown to converge to a
\emph{stationary point} (\ie, a point where the gradient vanishes) of the log-likelihood
\citep{wu:1983}. Note that convergence to the maximum likelihood estimator cannot in general be
guaranteed due to possible presence of multiple stationary points.

When processing large data sets or data streams  however, the EM algorithm becomes impractical
due to the requirement that the whole data be available at each iteration of the
algorithm. For this reason, there has been a strong interest for online variants of the EM which make it possible to estimate the parameters of a latent data model without
storing the data. In this work, we consider online algorithms for latent data models with
independent observations. The dominant approach (see also
Section~\ref{sec:state-of-ze-art} below) to online EM-like estimation follows the method proposed by
\cite{titterington:1984} which consists in using a 
%%% denomination pas du tout standard
%weighted parameter-space 
stochastic approximation algorithm, where the parameters are updated after each new observation
using the gradient of the incomplete data likelihood weighted by the complete data Fisher
information matrix. This approach has been used, with some variations, in many different
applications (see, \eg, \citealp{chung:bohme:2005,liu:almhana:choulakian:mcgorman:2006}); a proof
of convergence was given by \cite{wang:zhao:2006}.

In this contribution, we propose a new online EM algorithm that sticks more closely to the
principles of the original (batch-mode) EM algorithm. In particular, each iteration of the proposed
algorithm is decomposed into two steps, where the first one is a stochastic approximation version of
the E-step aimed at incorporating the information brought by the newly available observation, and,
the second step consists in the maximisation program that appears in the M-step of the traditional
EM algorithm. In addition, the proposed algorithm does not rely on the complete data information
matrix, which has two important consequences: firstly, from a practical point of view, the
evaluation and inversion of the information matrix is no longer required, secondly, the convergence
of the procedure does not rely on the implicit assumption that the model is \emph{well-specified},
that is, that the data under consideration is actually generated by the model, for some unknown
value of the parameter. As a consequence, and in contrast to previous work, we provide an analysis
of the proposed algorithm also for the case where the observations are not assumed to follow the
fitted statistical model. This consideration is particularly relevant in the case
of %regression, or
conditional missing data models, a simple case of which is used as an illustration of the proposed
online EM algorithm. Finally, it is shown that, with the additional use of Polyak-Ruppert
averaging, the proposed approach converges to the stationary points of the limiting normalised
log-likelihood criterion (\ie, the Kullback-Leibler divergence between the marginal density of the
observations and the model pdf) at a rate which is optimal.

The paper is organised as follows: In Section~\ref{sec:algo}, we review the basics of the EM and
associated algorithms and introduce the proposed approach. The connections with other
existing methods are discussed at the end of Section~\ref{sec:onlineEM_def} and a simple example
of application is described in Section~\ref{sec:exp:pois}. Convergence results are stated in
Section~\ref{sec:conv}, first in term of consistency (Section~\ref{sec:conv:consist}) and then of
convergence rate (Section~\ref{sec:conv:rate}), with the corresponding proofs given in
Appendix~\ref{sec:proofs}. Finally in Section~\ref{sec:regmix}, the performance of this approach is
illustrated in the context of mixture of linear regressions.

%==================================================================
\section{Algorithm Derivation}
\label{sec:algo}
\subsection{EM Basics}
\label{sec:intro}
In this section, we review the key properties of the EM algorithm as introduced
by \cite{dempster:laird:rubin:1977}. The latent
variable statistical model postulates the existence of a non-observable or
latent variable $X$ distributed under $\fx(x; \theta)$ where
$\{ \fx(x; \theta), \theta \in \Theta \}$ denotes a parametric
family of probability density functions indexed by a parameter value $\theta
\in \Theta \subset \rset^\nbt$. The observation $Y$ is then viewed as a
deterministic function of $X$ which takes its values in the set $\Yset$. This latent variable mechanism provides a
unified framework for situations which includes missing data, censored
observations, noisily observed data, \dots \citep{dempster:laird:rubin:1977}.

We will denote by $\fy(y;\theta)$ the (observed) likelihood function induced by the latent data
model. In addition, the notations $\PE_\theta[\cdot]$ and $\PE_\theta[\cdot|Y]$ will be used to
denote, respectively, the expectation and conditional expectation under the model parameterised by
$\theta$. Likewise, let $\pi$ denote the probability density function of the observation $Y$, where
we stress again that we do not restrict ourselves to the case where $\pi(\cdot) =
\fy(\cdot;\theta^\star)$, for an unknown value $\theta^\star$ of the parameter. The notations $\PP_\pi$
and $\PE_\pi$ will be used to denote probability and expectation under the actual distribution of
the observation.
% We will see latter in Section~\ref{sec:conv} that
% this distinction is important as not all algorithms under consideration are equally robust to the
% possible miss-specification of the model.

Given $n$ independent and identically distributed observations $Y_{1:n} \eqdef (Y_1, \dots, Y_n)$, the maximum
likelihood estimator is defined as $\hat{\theta}_n \eqdef
\mathrm{argmax}_{\theta\in\Theta} \, n^{-1}\logl{Y_{1:n}}{\theta}$, where
\begin{equation}
  \label{eq:LikelihoodFunction}
  \logl{Y_{1:n}}{\theta} \eqdef \sum_{i=1}^n \logl{Y_i}{\theta} \eqsp.
\end{equation}
Note that we normalise the log-likelihood (by $n$) to ease the transition to the online setting where $n$ increases when new observations become available.

The EM algorithm is an iterative optimisation algorithm that maximises the above (normalised) log-likelihood function despite the possibly complicated form of $\fy$ resulting from the latent data model. Traditionally, each EM iteration is decomposed in two steps. The E-step
consists in evaluating the conditional expectation
\begin{equation}
  \label{eq:EM-auxiliary-function}
  \funcEM{Y_{1:n}}{\theta}{\htheta_k} \eqdef n^{-1} \sum_{i=1}^n \PE_{\htheta_k}\left[
  \lx( X_i; \theta) \big| Y_i \right]
\end{equation}
where $\htheta_k$ is the current estimate of $\theta$, after $k$ iterations of
the algorithm. In the M-step, the value of $\theta$ maximising
$\funcEM{Y_{1:n}}{\theta}{\htheta_k}$ is found. This yields the new parameter
estimate $\htheta_{k+1}$. This two step process is repeated until convergence. 
The essence of the EM algorithm is that increasing
$\funcEM{Y_{1:n}}{\theta}{\htheta_k}$ forces an increase of the log-likelihood
$\logl{Y_{1:n}}{\theta}$ \citep{dempster:laird:rubin:1977}.

For $\mu: \Theta \to \rset$ a differentiable function, denote by $\nabla_\theta \mu= (\partial \mu/
\partial \theta_1, \dots, \partial \mu/ \partial \theta_{d_\theta})^T$ the \emph{gradient} of
$\mu$. If $\mu$ is twice differentiable, we denote by $\nabla_\theta^2 \mu$ the \emph{Hessian
  matrix} which is a $d_\theta \times d_\theta$ matrix whose components are given by
$[\nabla^2_\theta \mu]_{i,j}= \frac{\partial^2 \mu}{\partial \theta_i \partial \theta_j}$, $1 \leq
i,j \leq d_\theta$.  Following \cite{lange:1995}, if the M-step of the EM algorithm is replaced by
a Newton update, one obtains, assuming some regularity, the following recursion
\begin{equation}
  \label{eq:gradientEM-algorithm-lange}
  \htheta_{k+1} = \htheta_k + \gamma_{k+1} \left[\FIMcond{Y_{1:n}}{\hat{\theta}_k}\right]^{-1} \sum_{i=1}^n \PE_{\hat{\theta}_k} \left[ \nabla_{\theta} \lx( X_i; \hat{\theta}_k) \big| Y_i \right] \eqsp ,
\end{equation}
where $\gamma_{k+1}$ is a step size ($\gamma_{k+1} = 1$ correspond to the actual Newton update) and
$\FIMcond{Y_{1:n}}{\theta} = n^{-1} \sum_{i=1}^n \FIMcond{Y_i}{\theta}$ with $\FIMcond{y}{\theta} =
- \PE_\theta \left[ \nabla_\theta^2 \lx( X; \theta) \big| Y=y \right]$.
Note that due to  the so-called Fisher's identity (see discussion of \citealp{dempster:laird:rubin:1977}), the gradient term indeed coincides with the (observed data) score function as
\begin{equation}
  \PE_\theta \left[ \nabla_\theta \lx( X;
  \theta) \big| Y \right] = \dly(Y;\theta) \eqsp .
  \label{eq:Fisher}
\end{equation}
The algorithm in~\eqref{eq:gradientEM-algorithm-lange} can be shown to be locally equivalent to the
EM algorithm at convergence \citep{lange:1995}. In practise, the step-size
$\gamma_{k+1}$ is often adjusted using line searches to ensure that the likelihood is indeed increased at each iteration. In addition, $\FIMcond{Y_{1:n}}{\theta}$ is
not necessarily a positive definite matrix or could be badly conditioned; therefore, some
adjustment of the weight matrix $\FIMcond{Y_{1:n}}{\theta}$ may be necessary to avoid numerical
problems.

A well-known modification of the Newton recursion consists in replacing $\FIMcond{Y_{1:n}}{\theta}$
in~\eqref{eq:gradientEM-algorithm-lange} by the Fisher Information Matrix (FIM) associated to a
\emph{complete} observation,
\begin{equation}
  \label{eq:completeobservationFIM}
  \FIMcomplete{}{\theta} \eqdef - \PE_\theta \left[ \nabla_\theta^2 \log \fx(X;\theta) \right] \eqsp .
\end{equation}
Under the mild assumption that the complete data model is regular, $\FIMcomplete{}{\theta}$ is
guaranteed to be positive definite.  This modified recursion, which is more robust, may be seen as
an approximation of the \emph{scoring method} \citep{mclachlan:krishnan:1997}, where the complete
data FIM is used in place of the actual (\emph{observed}) FIM
\begin{equation}
  \label{eq:FIM}
   \FIMincomplete{\theta}\eqdef - \PE_\theta \left[ \nabla_\theta^2 \log \fy(Y;\theta) \right] \eqsp ,
\end{equation}
despite the fact that, in general, $\FIMincomplete{\theta}$ and $\FIMcomplete{}{\theta}$ are different. $\FIMcomplete{}{\theta}$ usually also differs from
$\FIMcond{Y_{1:n}}{\theta}$, as $\FIMcond{Y_{1:n}}{\theta}$ converges, as $n$ tends to infinity, to
\begin{equation}
  \label{eq:definitionFIMcomplete}
  \FIMcomplete{\pi}{\theta} \eqdef - \PE_\pi\left[ \PE_\theta \left[ \nabla_\theta^2 \log \fx(X;\theta) \big| Y \right] \right] \eqsp ,
\end{equation}
which doesn't correspond to a Fisher information matrix in the complete data model, except when
$\pi$ coincides with $f(\cdot;\theta)$. In the particular case, however, where the complete data
model belongs to a canonical (or naturally parameterised) exponential family of distributions,
$\FIMcond{Y_{1:n}}{\theta}$ coincides with $\FIMcomplete{}{\theta}$ and thus does not depend on
$\pi$ anymore. Hence, except in some specific cases or if one assumes that the model is
well-specified (\ie, $\pi = \fy(\cdot;\theta^\star)$), the convergence behaviour of the recursion
in~\eqref{eq:gradientEM-algorithm-lange} will be different when $\FIMcomplete{}{\theta}$ is used
instead of $\FIMcond{Y_{1:n}}{\theta}$.

\subsection{Stochastic Gradient EM Algorithms}
\label{sec:state-of-ze-art}
Being able to perform online estimation means that the data must be run through only once, which
is obviously not possible with the vanilla EM algorithm. To overcome this difficulty we consider in
the sequel online algorithms which produce, at a fixed computational cost, an updated parameter
estimate $\htheta_{n}$ for each new observation $Y_n$. Note that in the online setting, the
iteration index (which was previously denoted by $k$) is identical to the observation index $n$
and we will use the latter when describing the algorithms. To our best knowledge, the first online
parameter estimation procedure for latent data models is due to \cite{titterington:1984} who
proposed to use a stochastic approximation version of the modified gradient recursion:
\begin{equation}
  \label{eq:recursiveEM-Titterington}
  \htheta_{n+1} = \htheta_n + \gamma_{n+1} \FIMcomplete[-1]{}{\htheta_n} \dly(Y_{n+1}; \htheta_n) \eqsp,
\end{equation}
where $\{ \gamma_n \}$ is a decreasing sequence of positive step sizes. One may also consider using a  stochastic approximation version of the original (Newton) recursion in~(\ref{eq:gradientEM-algorithm-lange}):
\begin{equation}
  \label{eq:recursiveEM-Lange}
  \htheta_{n+1} = \htheta_n + \gamma_{n+1} \FIMcomplete[-1]{\pi}{\htheta_n}  \dly(Y_{n+1}; \htheta_n) \eqsp .
\end{equation}
Note that \eqref{eq:recursiveEM-Lange} does not correspond to a practical algorithm as
$\FIMcomplete{\pi}{\htheta_n}$ is usually unknown, although it can be estimated, for instance, by
recursively averaging over the values of $\FIMcond{Y_n}{\htheta_n}$. As discussed above however,
this algorithm may be less robust than~\eqref{eq:recursiveEM-Titterington} because
$\FIMcond{Y_n}{\htheta_n}$ is (usually) not guaranteed to be positive definite. In the
following, we will refer to~\eqref{eq:recursiveEM-Titterington} as \emph{Titterington's online
  algorithm} and to~\eqref{eq:recursiveEM-Lange} as the \emph{online gradient algorithm} (in
reference to the title of the paper by \citealp{lange:1995}). Note that 
%Titterington's 
both of these algorithms are based on the stochastic gradient approach and bear
very little resemblance with the original EM algorithm.

\subsection{The Proposed Online EM Algorithm}
\label{sec:onlineEM_def}
We now consider an online approach which is more directly related to the
principle underpinning the EM algorithm. The basic idea is to replace the
expectation step by a stochastic approximation step, while keeping the
maximisation step unchanged. More precisely, at iteration $n$, consider the
function
\begin{equation}
  \label{eq:recursiveEM}
  \hat{Q}_{n+1}( \theta)  = \hat{Q}_{n}( \theta)
  + \gamma_{n+1} \left(\PE_{\htheta_n} \left[ \log \fx(X_{n+1}; \theta) \big| Y_{n+1} \right] - \hat{Q}_{n}(\theta) \right) \eqsp ,
\end{equation}
and set $\htheta_{n+1}$ as the maximum of the function $\theta \mapsto \hat{Q}_{n+1}(\theta)$ over
the feasible set $\Theta$. One important advantage of \eqref{eq:recursiveEM} compared to
\eqref{eq:recursiveEM-Titterington} is that it automatically satisfies the parameter constraints without requiring any further modification. 
In addition, \eqref{eq:recursiveEM} does not explicitly require the inversion of a $(\nbt \times \nbt)$ matrix. For further comparisons between both approaches, both practical and in terms of rate of convergence, we refer to the example of Section~\ref{sec:exp:pois} and to the analysis of Section~\ref{sec:conv}.

Of course, this algorithm is of practical interest only if it is possible to compute and maximise
$\hat{Q}_n(\theta)$ efficiently. In the following we focus on the case where the complete data
likelihood belongs to an \emph{exponential family} satisfying the following assumptions. Let $\pscal{\cdot}{\cdot}$ denotes the scalar product between two vectors of $\rset^d$ and $|\cdot|$ the associated norm.

\begin{assumption}
\label{assum:expon}
\begin{description}
\item[(a)] The complete data likelihood is of the form
\begin{equation} \label{eq:curved-ex} \fx(x; \theta) = h(x) \exp\left\{-\psi(\theta) +
    \pscal{\suffstat(x)}{\phi(\theta)}\right\} \eqsp .
  \end{equation}
\item[(b)] The function
\begin{equation} \label{eq:defcondexpsuffstat} \condexpsuffstat{y}{\theta}
  \eqdef \PE_\theta \left[ \suffstat(X) \big| Y=y \right] \eqsp ,
\end{equation}
is well defined for all $(y,\theta) \in \Yset \times \Theta$.
% (\ie $\PE_\theta\left[
%    |\suffstat(X)| \big| Y=y \right] < \infty$, where $|\cdot|$ refers to the norm associated to the scalar product)
\item[(c)] There exists a convex open subset $\calS \subset \mathbb{R}^d$, which is such that
\begin{itemize}
  \item for all $s \in \calS$,  $(y,\theta) \in \Yset \times \Theta$ and $\gamma \in [0,1)$, $(1-\gamma) s + \gamma \condexpsuffstat{y}{\theta} \in \calS$.
  \item for any $s \in \calS$, the function $\theta \mapsto \ls(s;\theta) \eqdef -\psi(\theta) +
  \pscal{s}{\phi(\theta)}$ has a unique global maximum over $\Theta$ denoted $\functheta{s}$,\ie\
  \begin{equation} \label{eq:definition-functheta}
  \functheta{s} \eqdef \mathrm{argmax}_{\theta \in \Theta} \ls(s;\theta) \eqsp.
  \end{equation}
\end{itemize}

\end{description}
\end{assumption}

Assumption \ref{assum:expon} implies that the evaluation of $\PE_\theta \left[ \log \fx(X; \theta)
  \big| Y\right]$, and hence the E-step of the EM algorithm, reduces to the computation of the
expected value $\PE_\theta \left[ \suffstat(X) \big| Y \right]$ of the \emph{complete data
  sufficient statistic} $S(X)$. Indeed, the EM reestimation functional
$\funcEM{Y_{1:n}}{\theta}{\theta'}$ is then defined by
\begin{equation*}
    %\label{eq:defQEM}
    \funcEM{Y_{1:n}}{\theta}{\theta'} = \ls \left(
    n^{-1} \sum_{i=1}^n \condexpsuffstat{Y_i}{\theta'}; \theta \right) \eqsp .
\end{equation*}
The $(k+1)$-th iteration of the (batch mode) EM algorithm may thus be expressed as
\begin{equation} \label{eq:EM-recursion} \htheta_{k+1}= \functheta{ n^{-1}
    \sum_{i=1}^n \condexpsuffstat{Y_i}{\htheta_k}}\eqsp ,
\end{equation}
where the M-step corresponds to the application of the function $\bar\theta$. Note that the
construction of the set $\calS$ in Assumption \ref{assum:expon}(c) reflects the fact that in most
applications of EM, the M-step is unambiguous only when a sufficient number of observations have
been gathered. This point will be illustrated in the example to be considered in
Section~\ref{sec:regmix} below. Assumption \ref{assum:expon}(c) takes care of this issue in the case
of the online EM algorithm. As an additional comment about Assumption \ref{assum:expon}, note that
we do not require that $\phi$ be a one to one mapping and hence the complete data model may also
correspond to a \emph{curved} exponential family, where typically $\theta$ is of much lower
dimension than $\psi(\theta)$ (see, for instance, \cite{chung:bohme:2005,cappe:charbit:moulines:2006} for an
example involving Gaussian densities with structured covariance matrices).

In this setting, the proposed online EM algorithm takes the following form
\begin{align}
  \hcondexpsuffstat_{n+1} & = \hcondexpsuffstat_n + \gamma_{n+1}(\condexpsuffstat{Y_{n+1}}{\htheta_n} -\hcondexpsuffstat_n) \eqsp , \nonumber \\
   \htheta_{n+1} & = \functheta{\hcondexpsuffstat_{n+1}} \eqsp .
  \label{eq:recursiveEM-curvedexponentialfamily}
\end{align}

Algorithm of that kind have a rather long history in the machine learning
community.  The idea of sequentially updating the vector of sufficient
statistics has apparently been first proposed by \cite{nowlan:1991}, using a
fixed step size (or learning rate) $\gamma_n = \gamma$ (see also
\citealp{jordan:jacobs:1994}). The online EM algorithm
\eqref{eq:recursiveEM-curvedexponentialfamily} is also closely related to the
``incremental'' version of the EM algorithm derived by \cite{neal:hinton:1999}. The
incremental setting is more general than the recursive setting considered here,
because the observations are not necessarily processed sequentially in time and
might be used several times. The incremental EM algorithm of \cite{neal:hinton:1999}
defines the $(k+1)$-th parameter estimate as
\begin{equation} \label{eq:IncrmentalEM-recursion}
    \htheta_{k+1}= \functheta{ \left[\min(k+1,n)\right]^{-1}
    \sum_{i=1}^{\min(k+1,n)}  \hcondexpsuffstat_{k+1,i} } \eqsp,
\end{equation}
where $\hcondexpsuffstat_{k+1,i} = \hcondexpsuffstat_{k,i}$ if $i \neq I_{k+1}$ and
$\hcondexpsuffstat_{k+1,I_{k+1}} = \condexpsuffstat{Y_{I_{k+1}}}{\htheta_k}$. The index $I_{k+1}$
is typically chosen as $k+1$ while $k \leq n-1$ and runs through the data set, that is, $I_k \in \{1, \dots, n\}$, in a fixed or pseudo
random scanning order for subsequent iterations. When used in
batch mode (that is when $k > n$) it is seen that it mostly differs from the traditional EM
strategy in~\eqref{eq:EM-recursion} by the fact that the parameters are updated after each
computation of the conditional expectation of the complete data sufficient statistic corresponding to
one observation. When used in online mode ($k \leq n$), the algorithm of \cite{neal:hinton:1999}
coincides with the proposed online EM algorithm with a step-size of $\gamma_k = 1/k$ (see Section~\ref{sec:conv}  for further discussion of this particular choice of step sizes).
% We will show
% latter (in Section~\ref{sec:conv}) that this choice of stepsize decrease is indeed not appropriate
% in online mode as it is too fast to ensure consistent estimation of the parameters.
A specific
instance of the proposed online EM algorithm has been derived by \cite{sato:ishii:2000} for maximum
likelihood estimation in the so-called normalised Gaussian network; this algorithm was later extended by
\cite{sato:2000} to a canonical exponential family ($\phi(\theta)= \theta$ in \eqref{eq:curved-ex})
and a sketch of the proof of convergence, based on stochastic approximation results, was given. The
online EM algorithm defined in \eqref{eq:recursiveEM-curvedexponentialfamily} may be seen as a
generalisation of this scheme.

\subsection{An Example: Poisson Mixture}
\label{sec:exp:pois}
Before analysing the convergence of the above algorithm, we first consider a
simple example of application borrowed from \cite{liu:almhana:choulakian:mcgorman:2006}: consider the case of a mixture of $m$ Poisson distributions
\begin{equation}
  \label{eq:PoissonMixture:IncompleteLikelihood}
  \fy(y; \theta)= \sum_{j=1}^m \mixtureweight{j} \frac{\mixtureparam{j}^y}{y !} \rme^{-\mixtureparam{j}} \eqsp, \quad \text{for $y=0,1,2, \dots$} \eqsp,
\end{equation}
where the unknown parameters $\theta= (\mixtureweight{1}, \dots, \mixtureweight{m},
\mixtureparam{1}, \dots, \mixtureparam{m})$ satisfies the constraints $\mixtureweight{j} > 0$,
$\sum_{i=1}^m \mixtureweight{j} = 1$ and $\mixtureparam{j} > 0$.  In the mixture problem, the
incompleteness is caused by the ignorance of the component of the mixture.  Let $W$ be a random
variable taking value in $\{1, \dots, m\}$ with probabilities $\{ \mixtureweight{1}, \dots,
\mixtureweight{m} \}$. The random variable $W$ is called the regime or state and is not observable.
The probability density defined in \eqref{eq:PoissonMixture:IncompleteLikelihood} corresponds to
the assumption that $Y$ is distributed, given that $W=j$, according to a Poisson random variable
with parameter $\mixtureparam{j}$. Note that in this case, as in all examples which involve the
simpler missing data mechanism rather than the general latent data model introduced in
Section~\ref{sec:intro}, the complete data $X$ simply consists of the couple $(Y,W)$ and hence
conditional expectations of $X$ given $Y$ really boils down to expectations of $W$ given $Y$.

For the Poisson mixture, the complete data log-likelihood is given by
\begin{equation}
  \label{eq:PoissonMixture:CompleteLikelihood}
  \log \fx(y,w;\theta)
  = - \log(y!) + \sum_{j=1}^m \left[ \log (\mixtureweight{j}) -
    \mixtureparam{j} \right] \delta_{w,j} + \sum_{j=1}^m
  \log(\mixtureparam{j}) y \delta_{w,j} \eqsp,
\end{equation}
where $\delta_{i,l}$ is the Kronecker delta symbol: $\delta_{i,l} = 1$ if $i=l$ and $\delta_{i,l}=
0$ otherwise. The complete data likelihood may be rewritten as in \eqref{eq:curved-ex} with
$h(y,w)= - \log( y!)$, $\suffstat(y,w) = (\suffstat_1(y,w), \dots, \suffstat_m(y,w))$ and
$\phi(\theta) = (\phi_1(\theta), \dots, \phi_m(\theta))$, where
\begin{equation*}
  % \label{eq:sufficient-statistics-poisson}
  \suffstat_j(y,w) \eqdef \begin{pmatrix} \delta_{w,j} \\  y \delta_{w,j} \end{pmatrix} \eqsp, \qquad \text{and} \quad
  % \label{eq:function-poisson}
  \phi_j(\theta) \eqdef  \begin{pmatrix}
    \log(\mixtureweight{j}) - \mixtureparam{j} \\ \log(\mixtureparam{j})
  \end{pmatrix} \eqsp .
\end{equation*}
In this case, the conditional expectation of the
complete data sufficient statistics is fully determined by the
posterior probabilities of the mixture components defined by
\begin{equation*}
  \bmixtureind{j}(y;\theta) \eqdef \PP_\theta[W=j|Y=y] = \frac{\mixtureweight{j} \mixtureparam{j}^y \rme^{-\mixtureparam{j}}}{\sum_{l=1}^m \mixtureweight{l} \mixtureparam{l}^y \rme^{-\mixtureparam{l}}} \eqsp , \quad \text{for $j=1, \dots, m$} \eqsp .
\end{equation*}

The $(n+1)$-th step of the online EM algorithm consists in computing, for $j=1, \dots,
m$,
\begin{align}
  \hcondexpsuffstat_{j,n+1} & = \hcondexpsuffstat_{j,n} + \gamma_{n+1}
  \left\{ \begin{pmatrix} \bmixtureind{j}(Y_{n+1};\htheta_n) \\
      \bmixtureind{j}(Y_{n+1};\htheta_n) Y_{n+1}
    \end{pmatrix} - \hcondexpsuffstat_{j,n} \right\} \eqsp , \nonumber \\
  \hmixtureweight{j,n+1} & = \hcondexpsuffstat_{j,n+1}(1) \eqsp ,
  \qquad \hmixtureparam{j,n+1} =
  \frac{\hcondexpsuffstat_{j,n+1}(2)}{\hcondexpsuffstat_{j,n+1}(1)} \eqsp .
  \label{eq:REM-mixture-Poisson}
\end{align}
Comparing with the generic update equations in (\ref{eq:recursiveEM-curvedexponentialfamily}), one recognises the stochastic approximation version of the E-step, in the first line of~\eqref{eq:REM-mixture-Poisson}, followed by the application of $\bar\theta$.

To compare with Titterington's online algorithm in
\eqref{eq:recursiveEM-Titterington}, one first need to evaluate the complete
Fisher information matrix $\FIMcomplete{}{\theta}$. To deal with the equality
constraint $\sum_{j=1}^m \mixtureweight{j} = 1$, only the first $m-1$ weights
are used as parameters and the remaining one is represented as
$\mixtureweight{m} = 1 - \sum_{j=1}^{m-1} \mixtureweight{j}$ as in
\cite{titterington:1984}. The complete data Fisher information matrix defined
in~\eqref{eq:completeobservationFIM} is then given by
\[
\FIMcomplete{}{\theta} = \begin{pmatrix}
  \operatorname{diag}(\mixtureweight{1}^{-1}, \dots, \mixtureweight{m-1}^{-1}) + \mixtureweight{m}^{-1}\mathbf{1}_{m-1}\mathbf{1}_{m-1}^T  & \mathbf{0}_{(m-1) \times m} \\
  \mathbf{0}_{m \times (m-1)} &
  \operatorname{diag}(\mixtureweight{1}/\mixtureparam{1}, \dots,
  \mixtureweight{m}/\mixtureparam{m})
\end{pmatrix} \eqsp ,
\]
where the superscript $T$ denotes transposition, $\mathbf{1}$ and $\mathbf{0}$
respectively denote a vector of ones and a matrix of zeros, whose dimensions
are specified as subscript. Upon inverting $\FIMcomplete{}{\theta}$, the
following expression for the $(n+1)$-th step of Titterington's online
algorithm is obtained~:
\begin{align}
  \hmixtureweight{j,n+1} &= \hmixtureweight{j,n} + \gamma_{n+1} \left(
    \bmixtureind{j}(Y_{n+1},\htheta_n) - \hmixtureweight{j,n}
  \right) \eqsp, \nonumber \\
  \hmixtureparam{j,n+1} & = \hmixtureparam{j,n} + \gamma_{n+1}
  \frac{\bmixtureind{j}(Y_{n+1}; \htheta_n)}{\hmixtureweight{j,n}} \left(
    Y_{n+1} - \hmixtureparam{j,n} \right) \eqsp .
  \label{eq:REM-mixture-Poisson-Titterington}
\end{align}
To make the connection more explicit with the update of the online EM algorithm, note that due to
the fact that, in this simple case, there is an identification between some components of the vector
of sufficient statistics and the weight parameters (\ie, $\functheta{s_{j,n}} = \mixtureweight{j}$),
it is possible to rewrite~\eqref{eq:REM-mixture-Poisson} in terms of the latter only:
\begin{align*}
  \nonumber \hmixtureweight{j,n+1} & = \hmixtureweight{j,n} + \gamma_{n+1}
  \left( \bmixtureind{j}(Y_{n+1};\htheta_n) - \hmixtureweight{j,n}
  \right) \eqsp , \\
  \nonumber \hmixtureparam{j,n+1} & = \frac{
    \hmixtureparam{j,n}\hmixtureweight{j,n} + \gamma_{n+1} \left(
      \bmixtureind{j}(Y_{n+1};\htheta_n)Y_{n+1} -
      \hmixtureparam{j,n}\hmixtureweight{j,n}\right)}{\hmixtureweight{j,n+1}}
  \eqsp .
\end{align*}
In the Poisson mixture example, the two algorithms differ only in the way the intensities of the
Poisson components are updated. Whereas the online EM algorithm in~\eqref{eq:REM-mixture-Poisson}
does ensure that all parameter constraints are satisfied, it may happen, in contrast,
that~\eqref{eq:REM-mixture-Poisson-Titterington} produces negatives values for the intensities.
Near convergence however, the two algorithms behave very similarly in this simple case (see
Proposition~\ref{prop:asymptoticequivalencerecursiveEMs} below).

\subsection{Extensions}
\label{sec:exts}
As previously mentioned, \cite{neal:hinton:1999} advocate the use of online algorithms also in the
case of batch training with large sample sizes. The online algorithm then operates by repeatedly
scanning through the available sample. In our setting, this use of the proposed online EM algorithm may be
analysed by letting $\pi$ denote the empirical measure associated with the fixed sample $X_1,
\dots, X_n$. The results to follow thus also apply in this context, at least when the data scanning
order is random.

In semi-parametric regression models, each observation $Y$ comes with a vector of covariates $Z$
whose distribution is usually unspecified and treated as a nuisance parameter. To handle latent
data versions of regression models (mixture of regressions, mixture of experts, etc.---see
\citealp{gruen:leisch:2007,jordan:jacobs:1994}, as well as the example of Section~\ref{sec:regmix})
in our framework, one only needs to assume that the model consists of a parametric family
$\{f(x|z;\theta), \theta\in\Theta)\}$ of \emph{conditional} pdfs. In this setting however, it is
not possible anymore to compute expectations under the complete data distribution and the model can
never be well-specified, as the distribution of $Z$ is left unspecified. Thus Titterington's
algorithm in (\ref{eq:recursiveEM-Titterington}) does not directly apply in this setting. In
contrast, the proposed algorithm straightforwardly extends to this case by considering
covariate-dependent expectations of the sufficient statistics of $f(x|z;\theta)$, of the form
$\condexpsuffstat{y,z}{\theta} = \PE_\theta[\suffstat(X)|Y=y,Z=z]$, instead of
(\ref{eq:defcondexpsuffstat}). For notational simplicity, we state our results in the following
section without assuming the presence of covariates but extension to the case where there are
covariates is straightforward; the example of Section~\ref{sec:regmix} corresponds to a case where
covariates are available.

%==================================================================
\section{Convergence Issues}
\label{sec:conv}
\subsection{Consistency}
\label{sec:conv:consist}
In this section, we establish the convergence of the proposed algorithm towards the set of
stationary points of the Kullback-Leibler divergence between the actual observation density and the
model likelihood. These results are the analogues of those given by \cite{wang:zhao:2006} for
Titterington's online algorithm, with a somewhat broader scope since we do not assume that the
model is well-specified. The proofs corresponding to this section are given in Appendix
\ref{sec:proofs}. In addition to the conditions listed in Assumption~\ref{assum:expon}, we will require the following additional regularity assumptions.

\begin{assumption}
\label{assum:reg}
\begin{description}
\item[(a)] The parameter space $\Theta$ is a convex open subset of $\rset^\nbt$ and $\psi$ and $\phi$ in~(\ref{eq:curved-ex}) are twice continuously differentiable on $\Theta$.
\item[(b)] The function $s \mapsto \functheta{s}$, defined in \eqref{eq:definition-functheta}, is continuously differentiable on $\calS$,
\item[(c)] For some $p > 2$, and all compact subsets $\calK \subset \calS$,
  \[
    \sup_{s\in\calK} \PE_\pi\left( \left| \condexpsuffstat{Y}{\functheta{s}} \right|^p \right) < \infty \eqsp .
  \]
\end{description}
\end{assumption}

To analyse the recursion \eqref{eq:recursiveEM-curvedexponentialfamily}, the first step consists in expressing it as a standard Robbins-Monro stochastic approximation procedure operating on the complete data sufficient statistics:
\begin{equation}
  \label{eq:recursiveEM-curvedexponential-RM}
  \hcondexpsuffstat_{n+1} = \hcondexpsuffstat_n + \gamma_{n+1} \left(\meanfield(\hcondexpsuffstat_n) + \xi_{n+1}\right) \eqsp ,
\end{equation}
where $\meanfield: \calS \to \rset^{\nbt}$ is the so-called \emph{mean field} given by
\begin{equation}
  \label{eq:recursiveEM-curvedexponential-meanfield}
  \meanfield(s) \eqdef \PE_\pi \left[ \condexpsuffstat{Y}{\functheta{s}} \right] - s  \eqsp,
\end{equation}
and $\{ \xi_n \}_{n \geq 1}$ is a sequence of random variables representing stochastic perturbations defined by
\begin{align}
  \label{eq:recursiveEM-curvedexponentialfamily-noise}
  \xi_{n+1} %&\eqdef \condexpsuffstat{Y_{n+1}}{\functheta{\hcondexpsuffstat_{n}}} - \PE_\pi \left[ \condexpsuffstat{Y}{\functheta{\hcondexpsuffstat_n}} \right] \\
  &\eqdef \condexpsuffstat{Y_{n+1}}{\functheta{\hcondexpsuffstat_{n}}} -
  \CPE{\condexpsuffstat{Y_{n+1}}{\functheta{\hcondexpsuffstat_{n}}}}{\mathcal{F}_n} \eqsp ,
\end{align}
where $\mathcal{F}_n$ is the $\sigma$-field generated by $ \left(\hcondexpsuffstat_0, \{ Y_i \}_{i=1}^n \right)$.
%In the sequel, for any
%$\pi$-integrapble function $\varphi$, we denote by $\pi(\varphi)= \int \phi(y)
%\pi(y) \lambda_y(dy)$.
The aim of the Robbins-Monro procedure
\eqref{eq:recursiveEM-curvedexponential-RM} is to solve the equation
$\meanfield(s)= 0$.
As a preliminary step, we first characterise the set of roots of the mean field $\meanfield$. The
following proposition shows that, if $s^\star$ belongs to
\begin{equation}
  \label{eq:definition-gamma}
  \Gamma \eqdef \left\{s \in \calS: \meanfield(s)= 0  \right\} \eqsp,
\end{equation}
then $\theta^\star = \functheta{s^\star}$ is a stationary point of the
Kullback-Leibler divergence between $\pi$ and $\fy_\theta$,
\begin{equation}
  \label{eq:KullbackLeiblerDivergence}
  \kullback{\pi}{\fy_\theta} \eqdef %- \pi \left\{ \log
                                                   %\frac{\fy_\theta}{\pi}
                                                   %\right\}
\PE_\pi \left[ \log \left( \frac{\pi(Y)}{\fy(Y;\theta)} \right)  \right] \eqsp.
\end{equation}
\begin{proposition}
  \label{prop:characterization-stationary-points}
  Under Assumptions~\ref{assum:expon}--\ref{assum:reg},
  if $s^\star \in \calS$ is a root of $\meanfield$, \ie,
  $\meanfield(s^\star)=0$, then $\theta^\star = \functheta{s^\star}$
  %(where $\functheta{\cdot}$ is defined in \eqref{eq:definition-functheta})
  is a stationary point of the function $\theta \mapsto \kullback{\pi}{\fy_\theta}$, \ie, $\nabla_\theta \left.
    \kullback{\pi}{\fy_\theta} \right|_{\theta= \theta^\star}= 0$.
  Conversely, if $\theta^\star$ is a stationary point of $\theta \mapsto
  \kullback{\pi}{\fy_\theta}$, then $s^\star = \PE_\pi [\condexpsuffstat{Y}{\theta^\star}]$
  % (where $\condexpsuffstat{\cdot}{\cdot}$  is defined in \eqref{eq:defcondexpsuffstat})
  % \int \condexpsuffstat{y}{\theta^\star} \pi(y) \lambda_y(dy)$
  is a root of $\meanfield$.
  % $\meanfield(s^\star)= 0$.
\end{proposition}

We then show that the
function $\lyapunov: \calS \to [0,\infty)$ defined by
\begin{equation}
  \label{eq:definition-lyapunov}
  \lyapunov(s) \eqdef \kullback{\pi}{\fy_{\functheta{s}}} \eqsp,
\end{equation}
is a \emph{Lyapunov function} for the mean field $\meanfield$ and the set $\Gamma$, \ie\ for any $s
\in \calS$, $\pscal{\nabla_s \lyapunov(s)}{\meanfield(s)} \leq 0$ and $\pscal{\nabla_s
  \lyapunov(s)}{\meanfield(s)} =0$ if and only if $\meanfield(s)= 0$. The existence of a Lyapunov
function is a standard argument to prove the global asymptotic stability of the solutions of the
Robbins-Monro procedure. This property can be seen as an analog of
the monotonicity property of the EM algorithm: each unperturbed iteration
$\condexpsuffstatletter_{k+1} = \condexpsuffstatletter_k + \gamma_{k+1}
\meanfield(\condexpsuffstatletter_k)$ decreases the Kullback-Leibler divergence to the target
distribution $\pi$, provided that $\gamma_{k+1}$ is small enough.
%Denote  by  $\Gamma \eqdef \left\{ s \in \calS, \meanfield(s) \right\}$ the set of roots of $\meanfield$.
\begin{proposition}
  \label{prop:lyapunov-function-h}
  Under Assumptions~\ref{assum:expon}--\ref{assum:reg},
  \begin{itemize}
  \item $\lyapunov(s)$ is continuously differentiable on $\calS$,
  \item for any compact subset $\mathcal{K} \subset \calS \setminus \Gamma$,
    $$
     \sup_{s \in \mathcal{K}} \pscal{\nabla_s \lyapunov(s)}{\meanfield(s)} < 0 \eqsp.
    $$
  \end{itemize}
\end{proposition}
Using this result, we may now prove the convergence of the sequence
$\{\hcondexpsuffstat_k\}$.
Denote by $\mathcal{L} = \left\{ \theta \in \Theta : \nabla_\theta \kullback{\pi}{\fy_\theta}= 0 \right\}$ the set of
stationary points of the Kullback-Leibler divergence, and, for $x \in \rset^m$ and $A \subset \rset^m$, let $d(x,A) = \inf \{ y
\in A, |x-y| \}$.

\begin{theorem}
  \label{theo:wp1convergencerecursiveEM}
  Assume~\ref{assum:expon}--\ref{assum:reg}
  and that, in addition,
  \begin{enumerate}
    \item $0< \gamma_i < 1$, $\sum_{i=1}^\infty \gamma_i = \infty$ and $\sum_{i=1}^\infty \gamma_i^2 < \infty$,
     \item $\hcondexpsuffstat_0 \in \calS$ and with probability one, $\limsup |\hcondexpsuffstat_n| < \infty$ and $\liminf d(\hcondexpsuffstat_n,\calS^c) > 0$.
     \item The set $\lyapunov(\Gamma)$ is nowhere dense.
  \end{enumerate}
  Then, $\lim_{n \to \infty} d(\hcondexpsuffstat_n,\Gamma)= 0$ and $\lim_{n \to
    \infty} d( \htheta_n, \mathcal{L}) = 0$, with probability one.
\end{theorem}

The first condition of Theorem~\ref{theo:wp1convergencerecursiveEM} is standard for
decreasing step-size stochastic approximation procedures \citep{kushner:yin:1997}. It is satisfied for instance by setting
$\gamma_i = \gamma_0 i^{-\alpha}$, with $\alpha \in (1/2,1]$. The additional requirements that $\gamma_i$ be less than 1 and $\hcondexpsuffstat_0$ be chosen in $\calS$ is just meant to ensure that the whole sequence $\{\hcondexpsuffstat_k\}$ stays in $\calS$ (see Assumption~\ref{assum:expon}(c)). The rest of the second assumption of Theorem~\ref{theo:wp1convergencerecursiveEM} correspond to a stability assumption which is not trivial.
%, except in some particular cases where $\calS$ is bounded. In the Poisson mixture example of Section~\ref{sec:exp:pois} for instance, the first component of $\hcondexpsuffstat_{j,k}$ is by definition in $(0,1)$ (if $\hcondexpsuffstat_{j,0}$ is) while the second component lies between 0 and $\max\{Y,\dots,Y_k\}$. Hence the stability assumption will be easily verified only if one assume that the distribution of the observed counts has a bounded support.
In general settings, the stability of the algorithm can be enforced by truncating the algorithm updates, either on a fixed set (see, \eg, \citealp[chapter 2]{kushner:yin:2003}) or on an expanding sequence of sets (see, \eg, \citealp[chapter 2]{chen:book:2002}, or
 \citealp{andrieu:moulines:priouret:2005}). We do not explicitly carry out these constructions  here to keep the exposition concise.

\subsection{Rate of Convergence}
\label{sec:conv:rate}
In this section, we show that when approaching convergence, the
online EM algorithm is comparable to the online gradient algorithm in~(\ref{eq:recursiveEM-Lange}).
The existence of such links is hardly surprising,
in view of the discussions in Section 4 of \cite{titterington:1984} and
Section 3 of \cite{lange:1995}, and may be seen as a counterpart, for stochastic
approximation, of the asymptotic equivalence of the gradient EM algorithm of \cite{lange:1995} and
the EM algorithm. To highlight these relations, we first express the online EM algorithm as a stochastic approximation procedure on $\theta$.
%and then show
%that this stochastic approximation is, up to a vanishingly small term, related to the recursive gradient EM algorithm.
\begin{proposition}
  \label{prop:asymptoticequivalencerecursiveEMs}
  Under the assumptions of Theorem \ref{theo:wp1convergencerecursiveEM}, the
  online EM sequence $\{\htheta_n \}_{n \geq 0}$ given by
  \eqref{eq:recursiveEM-curvedexponentialfamily} follows the recursion
  \begin{equation}
    \label{eq:recursiveEM-in-theta}
    \htheta_{n+1} = \htheta_n + \gamma_{n+1} \, \FIMcomplete[-1]{\pi}{\htheta_n}  \dly(Y_{n+1};\htheta_n) + \gamma_{n+1} \rho_{n+1}
  \end{equation}
  where $\lim_{n \to \infty} \rho_n = 0$ \as\ and $\FIMcomplete{\pi}{\theta}$ is defined in \eqref{eq:definitionFIMcomplete}.
\end{proposition}

Hence, the online EM algorithm is equivalent, when approaching convergence, to the online gradient
algorithm defined in (\ref{eq:recursiveEM-Lange}) which coincides with Titterington's online
algorithm with $\FIMcomplete{\pi}{\htheta_n}$ substituted for $\FIMcomplete{}{\htheta_n}$. It is
remarkable that the online EM algorithm can achieve a convergence performance similar to that of
the online gradient algorithm without explicit matrix approximation nor inversion. Note that, as previously discussed in Section~\ref{sec:intro}, in the particular case of canonical exponential families, $\FIMcomplete{\pi}{\theta}$ and $\FIMcomplete{}{\theta}$ coincide and the proposed online EM algorithm is thus also equivalent (near convergence) to Titterington's online algorithm.

Although the recursion \eqref{eq:recursiveEM-in-theta} will not lead to
asymptotic efficiency, we can, under appropriate additional conditions,
guarantee $\gamma_n^{-1/2}$-consistency and asymptotic normality. We use the
weak convergence result presented in \citealp[Theorem 1]{pelletier:1998}.
%which applies to the case where the set of limiting points is not reduced to a single element.

\begin{theorem}
  \label{theo:weak-convergence-REM}
  Under the assumptions of Theorem \ref{theo:wp1convergencerecursiveEM}, let $\theta^\star$ be a (possibly local) minimum of the Kullback-Leibler
  divergence: $\theta \mapsto \kullback{\pi}{\fy_\theta}$.  Denote by
  \begin{align*}
    &H(\theta^\star) \eqdef \FIMcomplete[-1]{\pi}{\theta^\star} \left[ - \nabla_\theta^2 \left. \kullback{\pi}{\fy_\theta} \right|_{\theta= \theta^\star} \right] \eqsp, \\
    &\Gamma(\theta^\star) \eqdef \FIMcomplete[-1]{\pi}{\theta^\star} \, \PE_\pi \left(
      \dly(Y;{\theta^\star}) \left\{\dly(Y;{\theta^\star})\right\}^T\right)
    % \left[ \int \left\{ \dly(y;\theta^\star) \right\}^{\otimes 2} \pi(y)
    %   \lambda_y(dy) \right]
    \, \FIMcomplete[-1]{\pi}{\theta^\star} \eqsp.
  \end{align*}
  Then,
  \begin{enumerate}
  \item \label{item:weak-convergence-REM:stable} $H(\theta^\star)$ is a stable
    matrix whose eigenvalues have their real part upper bounded by
    $-\lambda(\theta^\star)$, where $\lambda(\theta^\star) > 0$.
  \item \label{item:weak-convergence-REM:WK} Let $\gamma_n = \gamma_0
    n^{-\alpha}$, where $\gamma_0$ may be chosen freely in $(0,1)$ when $\alpha \in (1/2,1)$ but must satisfy
$\gamma_0 > \lambda(\theta^\star)$ when $\alpha = 1$; then, on the event
    $\Omega(\theta^\star) = \{ \lim_{n \to \infty} \htheta_n =
    \theta^\star\}$, the sequence $\gamma_n^{-1/2} \left( \htheta_n -
      \theta^\star \right)$ converges in distribution to a zero mean Gaussian
    distribution with covariance $\Sigma(\theta^\star)$, where
    $\Sigma(\theta^\star)$ is the solution of the Lyapunov equation
    \begin{equation}
      \label{eq:LyapunovEquation}
      \left( H(\theta^\star) + \zeta \mathrm{Id} \right) \Sigma(\theta^\star) +  \Sigma(\theta^\star) \left( H^T(\theta^\star) + \zeta \mathrm{Id} \right) = - \Gamma(\theta^\star) \eqsp,
    \end{equation}
    where $\zeta=0$ if $\alpha \in (1/2,1)$ and $\zeta=
    \gamma_0^{-1}$ if $\alpha = 1$, and, $\mathrm{Id}$ denotes the identity matrix.
  \end{enumerate}
\end{theorem}

Solving equation (28) is easy for a well-specified model, \ie, when $\pi=\fy_{\theta^\star}$, as
the FIM $\FIMincomplete{\theta^\star}$ that is associated with the (observed) data model then
satisfies
\begin{multline*}
  \FIMincomplete{\theta^\star} = -
\PE_{\theta^\star} \left[ \nabla_\theta^2 \log \fy(Y;{\theta^\star}) \right] \\
  = \nabla_\theta^2 \left.
  \kullback{\fy_{\theta^\star}}{\fy_\theta} \right|_{\theta= \theta^\star} =
\PE_\pi \left( \dly(Y;{\theta^\star}) \left\{\dly(Y;{\theta^\star})\right\}^T \right) \eqsp .
\end{multline*}
When $\zeta= 0$, the solution of the Lyapunov equation is thus given by $\Sigma(\theta^\star) =
\FIMcomplete[-1]{\pi}{\theta^\star}/2$. The model being well-specified also implies that
$\FIMcomplete{\pi}{\theta^\star} = \FIMcomplete{}{\theta^\star}$, and, hence, the asymptotic
covariance matrix is given by half the inverse of the complete data FIM in this case.
When $\zeta \ne 0$, the Lyapunov equation cannot be solved in
explicitly, except when the parameter is scalar (the result then coincides with
\citealp[Theorem 1]{titterington:1984}). Note that using $\gamma_n = \gamma_0 n^{-\alpha}$ with $\alpha
= 1$ provides the optimal convergence rate of $1/\sqrt{n}$ but only at the price of a constraint on
the scale $\gamma_0$, which is usually impossible to check in practice. On the other hand, using
$\alpha \in (1/2,1)$ results in a slower convergence rate but without constraint on the scale
$\gamma_0$ of the step-size (except for the fact that it has to be smaller than 1).

To circumvent this difficulty, we recommend to use the so-called Polyak-Ruppert averaging technique \citep{polyak:1990,ruppert:1988} as a post-processing step. Following \cite{polyak:1990} ---see
also \cite{polyak:juditsky:1992,mokakdem:pelletier:2006}---, if $\gamma= \gamma_0
n^{-\alpha}$, with $\alpha \in (1/2,1)$, then the running average
\begin{equation}
  \label{eq:SGA:averaging}
  \tilde{\theta}_n \eqdef n^{-1} \sum_{j=n_0}^n \htheta_j \eqsp, \qquad n \geq 1
\end{equation}
converges at rate $1/\sqrt{n}$, for all values of $\gamma_0$.
Furthermore, on the event $\Omega(\theta^\star)$ defined in Theorem~\ref{theo:weak-convergence-REM} above, $\sqrt{n}( \tilde{\theta}_n -
\theta^\star)$ is asymptotically normal, with asymptotic covariance matrix
\begin{multline}
  \label{eq:CovarianceMatrixAveraging}
  \overline{\Sigma}(\theta^\star)= H^{-1}(\theta^*) \Gamma(\theta^\star) H^{-1}(\theta^\star) = \\
  \left[ - \nabla_\theta^2 \left. \kullback{\pi}{\fy_\theta} \right|_{\theta=
      \theta^\star} \right]^{-1} \pi \left( \dly_{\theta^\star}
    \left\{\dly_{\theta^\star}\right\}^T \right) \left[ - \nabla_\theta^2 \left.
      \kullback{\pi}{\fy_\theta} \right|_{\theta= \theta^\star} \right]^{-1}
  \eqsp,
\end{multline}
which is known to be optimal \citep{kushner:yin:1997}. If $\pi =
\fy_{\theta^\star}$, the previous result shows that the averaged sequence
$\tilde{\theta}_n$ is an asymptotically efficient sequence of estimates of
$\theta^\star$, \ie\ the asymptotic covariance of $\sqrt{n} (\tilde{\theta}_n -
\theta^\star)$ is equal to the inverse of the (observed data) FIM $\FIMincomplete{\theta^\star}$.

%==================================================================
\section{Application to Mixtures of Gaussian Regressions}
\label{sec:regmix}
To illustrate the performance of the proposed method, we consider a regression model which, as discussed in Section~\ref{sec:exts}, corresponds to a case where the complete data FIM is not available. In contrast, we illustrate below the fact that the proposed algorithm, without explicitly requesting the determination of a weighting matrix does provide asymptotically efficient parameter estimates when Polyak-Ruppert averaging is used.

The model we consider is a finite mixture of Gaussian linear regressions, where the complete data
consists of the response variable $R$, here assumed to be scalar for simplicity, the
$\nbz$-dimensional vector $Z$ that contains the explanatory variables, and, $W$ which corresponds,
as in the example of Section~\ref{sec:exp:pois}, to a latent mixture indicator taking its value in
the finite set $\{1,\dots,m\}$. We assume that given $W=j$ and $Z$, $R$ is distributed as a
$\mathcal{N}(\beta_j^T Z, \sigma_j^2)$ Gaussian variable, while $W$ is independent of $Z$ and such
that $\PP_\theta(W = j) = \omega_j$. Thus the parameters of the model are the mixture weights
$\omega_j$ and the regression vectors $\beta_j$ and variances $\sigma_j^2$, for $j=1,\dots,m$. As
is usually the case in conditional regression models, we specify only the part of the complete data
likelihood that depends on the parameters, without explicitly modelling the marginal distribution
of the vector of regressors $Z$. In terms of our general notations, the complete data $X$ is the triple
$(R,Z,W)$, the observed data is the couple $(R,Z)$ and the model is not well-specified, in the
sense that the distribution of the observation $(R,Z)$ is not fully determined by the model.  We
refer to \cite{hurn:justel:robert:2003} or \cite{gruen:leisch:2007} and references therein for more
information on mixture of regression models and their practical use.

In the mixture of  Gaussian regressions model, the part of the complete data log-likelihood that depends on the parameters may be written as
\begin{equation}
  \label{eq:RegressionMixture:CompleteLikelihood}
  \log \fx(r,w,z;\theta)
  = \sum_{j=1}^m \left\{ \log (\mixtureweight{j}) - \frac12 \left[\log\sigma_j^2 + \frac{\left(r-\beta_j^T z\right)^2}{\sigma_j^2}\right] \right\} \delta_{w,j} \eqsp,
\end{equation}
where $\delta$ denotes, as before, the Kronecker delta. To
put~\eqref{eq:RegressionMixture:CompleteLikelihood} in the form given
in~~(\ref{eq:curved-ex}), one needs to define the statistics $S =
(S_{1,j},S_{2,j},S_{3,j},S_{4,j})_{1\leq j\leq m}$ where
\begin{align}
\label{eq:sufficient-statistics:regressionmixture}
 & S_{1,j}(r,w,z) = \delta_{w,j} & \quad & \text{(scalar)} \eqsp , \nonumber\\
 & S_{2,j}(r,w,z) = \delta_{w,j} r z & & \text{($\nbz \times 1$)} \eqsp , \nonumber \\
 & S_{3,j}(r,w,z) = \delta_{w,j} zz^T & & \text{($\nbz \times \nbz$)} \eqsp , \nonumber \\
 & S_{4,j}(r,w,z) = \delta_{w,j} r^2 & & \text{(scalar)} \eqsp .
\end{align}
As in the simple Poisson mixture example of Section~\ref{sec:exp:pois}, the E-step statistics only depend on the conditional expectation of the indicator variable $W$ through
\begin{align}
 & \condexpsuffstatletter_{1,j}(r,z;\theta) = \bar{w}_{j}(r,z;\theta) \eqsp , \nonumber\\
 & \condexpsuffstatletter_{2,j}(r,z;\theta) = \bar{w}_{j}(r,z;\theta) r z \eqsp , \nonumber \\
 & \condexpsuffstatletter_{3,j}(r,z;\theta) = \bar{w}_{j}(r,z;\theta) zz^T \eqsp , \nonumber \\
 & \condexpsuffstatletter_{4,j}(r,z;\theta) = \bar{w}_{j}(r,z;\theta) r^2 \eqsp ,
 \label{eq:E-step:regressionmixture}
\end{align}
where
\begin{equation*}
  \bar{w}_{j}(r,z;\theta) \eqdef \PE_\theta[W=j|R=r,Z=z] = \frac{\frac{\mixtureweight{j}}{\sigma_j} \exp\left[-\frac12 \frac{(r-\beta_j^T z)^2}{\sigma_j^2}\right]}{\sum_{l=1}^m \frac{\mixtureweight{l}}{\sigma_l} \exp\left[-\frac12 \frac{(r-\beta_l^T z)^2}{\sigma_l^2}\right]}  \eqsp .
\end{equation*}
Finally, it is easily checked that that the M-step is equivalent to an application of the
function $\bar\theta : s \mapsto \left(\bar\omega_j(s),\bar\beta_j(s),\bar\sigma_j(s)\right)_{1\leq j\leq m}$ where
\begin{align}
 & \bar\omega_j(s) = s_{1,j} \eqsp , \nonumber \\
 & \bar{\beta}_j(s) =  s_{3,j}^{-1} s_{2,j} \eqsp , \nonumber \\
 & \bar{\sigma}_j^2(s) = \left(s_{4,j} - \bar{\beta}_j^T(s) s_{2,j} \right) / s_{1,j} \eqsp \eqsp . \label{eq:M-step:regressionmixture}
\end{align}

In this example, the role played by the set $\calS$ in Assumption~\ref{assum:expon}(c) is important: In order to apply~\eqref{eq:M-step:regressionmixture}, its required that the scalars $s_{1,j}$ belong to the open set $(0,1)$ and that the $(\nbz+1)$-dimensional matrices block-defined by
\[
  M_j =
  \begin{pmatrix}
    s_{3,j} & s_{2,j} \\
    s_{2,j}^T & s_{4,j} \\
  \end{pmatrix} \eqsp ,
\]
be positive definite, since $\bar{\sigma}_j^2(s)$ is, up to normalisation by $s_{1,j}$, the Schur complement of $M_j$. These constraints, for $j=1,\dots,m$ define the set $\calS$ which is indeed open and convex. The function $\condexpsuffstatletter$ defined in~\eqref{eq:E-step:regressionmixture} however \emph{never} produces values of $s$ which are in $\calS$. In particular, $\condexpsuffstatletter_{3,j}(r,z;\theta)$ is a rank one matrix which is not invertible (unless $\nbz=1$). Hence the importance of using an initialisation $s_0$ which is chosen in $\calS$. For the simulations below, we took care of this issue by inhibiting the parameter re-estimation step in~\eqref{eq:M-step:regressionmixture} for the first twenty observations of each run. In other words, the first twenty observations are used only to build a up a value of $\hcondexpsuffstat_{20}$, using the first line of~\eqref{eq:recursiveEM-curvedexponentialfamily},  which is in  $\calS$ with great probability.

For illustration purpose, we consider a variation of a simple simulation example used in the \verb+flexmix+ \verb+R+ package \citep{leisch:2004}, where $m = 2$, $\omega_1 = \omega_2 = 0.5$, and
\begin{equation*}
  R =
  \begin{cases}
     5 U + V & \text{(when $W = 1$)} \\
     15 + 10 \, U - U^2 + V & \text{(when $W = 2$)} \\
  \end{cases} \eqsp ,
\end{equation*}
where $U \sim \operatorname{Unif}(0,10)$ and $V \sim \mathcal{N}(0,9^2)$. In order to balance the asymptotic variances of the regression parameters (see below) we used $Z^T = (1, U,
U^2/10)$ as the vector of regressors, hence the actual value of the regression parameter is $\beta_1^T = (0, 5, 0)$ for the first component and $\beta_2^T = (15, 10, -10)$. The corresponding data is shown in Figure~\ref{fig:batch_data} where the points corresponding to both classes are plotted differently for illustration purpose, despite the fact that only unsupervised estimation is considered here. The labelling is indeed rather ambiguous in this case as the posterior probability of belonging to one of the two classes is between 0.25 and 0.75 for about 40\% of the points.

\begin{figure}[hbtp] \centering
  \includegraphics[width=10cm]{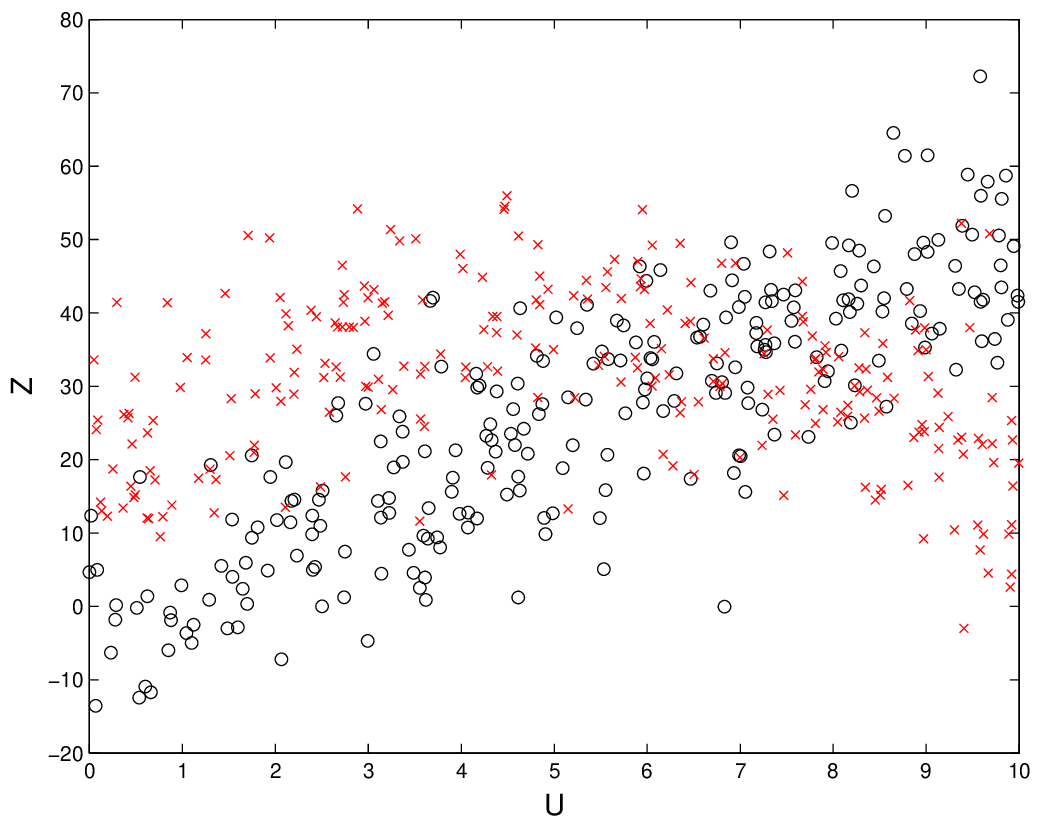}
  \caption{500 points drawn from the model: circles, points drawn from the first class, and, crosses, points drawn from the second class (the algorithms discussed here ignore the class labels).}
  \label{fig:batch_data}
\end{figure}

Clearly, the mixture of regressions model is such that the associated complete data sufficient
likelihood has the form given in Assumption~\ref{assum:expon}(a), where the marginal density of the
explanatory variables in $Z$ appears in the term $h(x)$ since it does not depend on the parameters.
Hence the previous theory applies straightforwardly and the online EM algorithm may be used to
maximise the conditional likelihood function of the responses $R$ given the regressors $Z$.
However, the explicit evaluation of the complete data FIM $\FIMcomplete{}{\theta}$ defined
in~\eqref{eq:completeobservationFIM} is not an option here because the model does not specify the
marginal distribution of $Z$. Titterington's online algorithm may thus not be used directly.
Applying the recursion in (\ref{eq:recursiveEM-Titterington}) without a weighting matrix is not recommended here as the regression parameters are greatly correlated due to the
non-orthogonality of the regressors.

\begin{figure}[p] \centering
  \includegraphics[width=10cm]{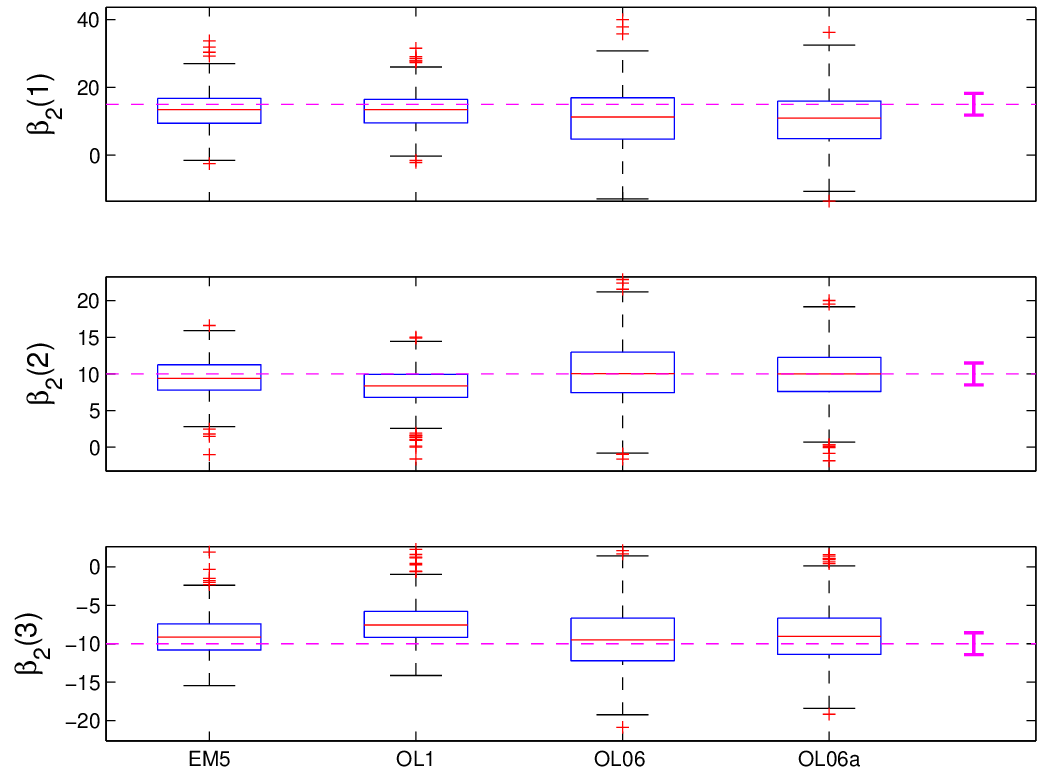}
  \caption{Box-and-whisker plots of the three components of $\beta_2$ (from top to bottom)
    estimated from 500 independent runs of length $n=100$ for, EM5: five iterations of the batch EM
    algorithm, OL1: online EM algorithm with $\gamma_i = 1/i$, OL06: online EM algorithm with
    $\gamma_i = 1/i^{0.6}$, OL06a: online EM algorithm with $\gamma_i = 1/i^{0.6}$ and averaging
    started from the 50th iteration. The horizontal dashed line corresponds to the actual parameter
    value and the interval in bold at the right of each plot to the interquartile range
    deduced from the asymptotic normal approximation.}
  \label{fig:batch_n100}
\end{figure}

\begin{figure}[p] \centering
  \includegraphics[width=10cm]{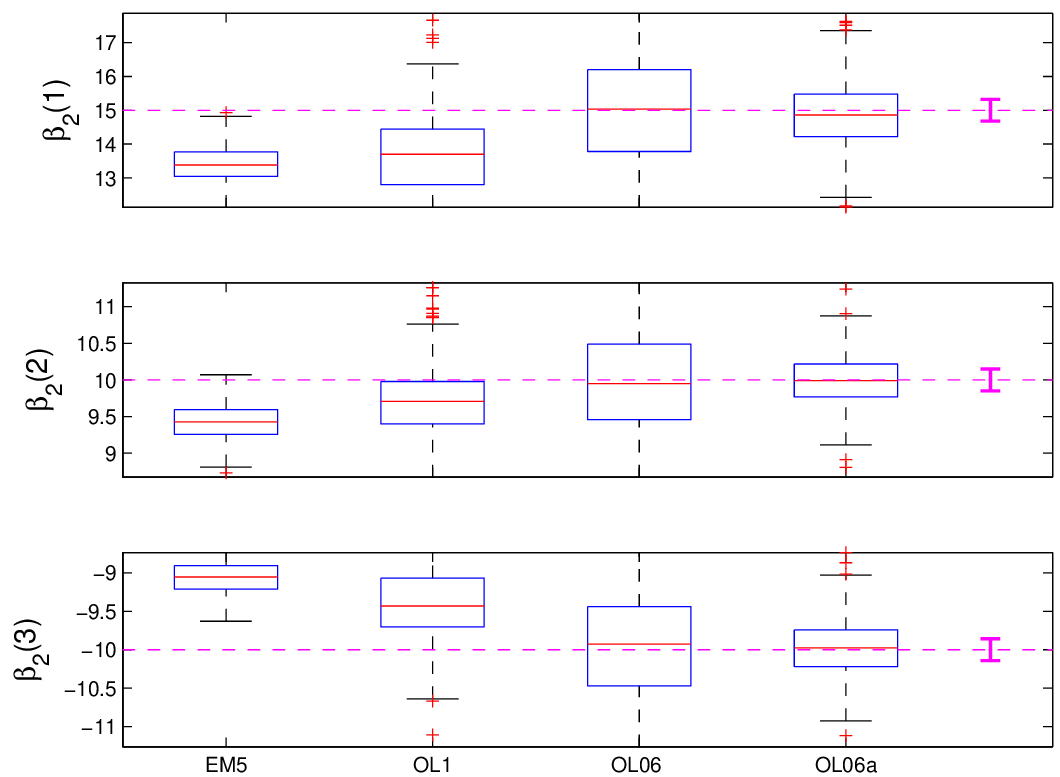}
  \caption{Same plots as in Figure~\ref{fig:batch_n100} for signals of length $n = 10,000$ (OL06a uses averaging started from the $5,000$th iteration).}
  \label{fig:batch_n10000}
\end{figure}

In order to determine a suitable weighting matrix, one can use Fisher's relation~\eqref{eq:Fisher}
which gives, for the regression parameters,
\begin{multline*}
  \nabla_{\beta_j} \log g(r|z;\theta) = \PE_\theta \left[\left. \nabla_{\beta_j} \log f(R,W|Z;\theta) \right| R=r, Z=z; \theta\right] \\
   = \PE_\theta \left[\left. \delta_{W,j} \frac{(R-\beta_j^T Z) Z}{\sigma_j^2} \right| R=r, Z=z; \theta\right] = \bar{w}_{j}(r,z;\theta) \frac{(r-\beta_j^T z) z}{\sigma_j^2} \eqsp .
  \label{eq:regmix:Fisher}
\end{multline*}
Hence, if we assume that the model is well-specified, the (observed) FIM $\FIMincomplete{\theta}$ may be approximated, near convergence, by computing
empirical averages of the form
\[
  1/n \sum_{i=1}^n \nabla_{\beta_j} \log g(R_i|Z_i;\theta) \left\{\nabla_{\beta_j}
  \log g(R_i|Z_i;\theta) \right\}^T \eqsp .
\]
As the online EM algorithm does not require such
computations, this estimate has been used only to determine the FIM at the actual parameter value for
comparison purpose. It is easily checked that due to the linearity of the model and the fact that
both components have equal weights and variance, the covariance matrices for $\beta_1$ and
$\beta_2$ are the same. The numerical approximation determined from a million simulated
observations yields asymptotic standard deviations of $(47.8, 22.1, 21.1)$ for the coordinates of
$\beta_j$, with an associated correlation matrix of
\[
\begin{pmatrix}
  1 & -0.87 & 0.75 \\
  -0.87 & 1 & -0.97 \\
  0.75 & -0.97 & 1
\end{pmatrix} \eqsp .
\]
As noted above, the coordinates of the regression vector are very correlated which would make the unweighted parameter-space stochastic approximation algorithm (\ie, \eqref{eq:recursiveEM-Titterington} with an identity matrix instead of $\FIMcomplete[-1]{}{\htheta_n}$) very inefficient.

\begin{figure}[hbtp] \centering
  \includegraphics[width=10cm]{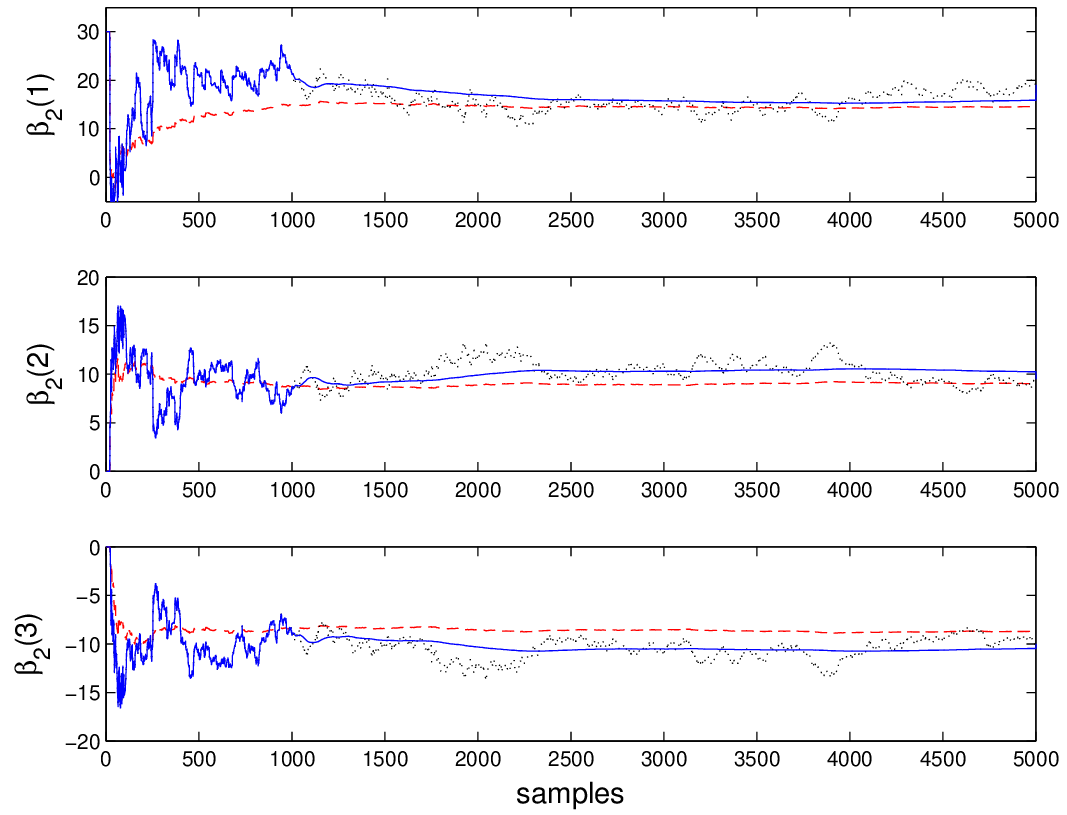}
  \caption{Example of parameters trajectories for the three components of $\beta_2$ (from top to bottom) for a signal of length $n = 5,000$: OL1, dashed line, OL06 doted line, OL06a solid line (with averaging started after $1,000$ iterations.}
  \label{fig:batch_traj}
\end{figure}

For run-lengths of $n=100$ and $n=10,000$ observations, we illustrate the performance of the following four algorithmic options~:
\begin{description}
\item[EM5] Five iterations of the batch EM algorithm, using the whole data.
\item[OL1] The online EM algorithm with step size $\gamma_n = 1/n$.
\item[OL06] The online EM algorithm with step size $\gamma_n = n^{-0.6}$. 
\item[OL06a] The online EM algorithm with step size $\gamma_n = n^{-0.6}$ and averaging started from $n_0 = n/2$ according to (\ref{eq:SGA:averaging}). 
\end{description}
Note that whereas OL1 and OL06a have the same computational complexity (as the averaging
post-processing has a negligible cost), EM5 is significantly more costly requiring five times as
many E-step computations; it is also non-recursive. All algorithm are started from the same point
and run for 500 independent simulated replicas of the data. The results (for $\beta_2$) are
summarised as box-and-whisker plots in Figure~\ref{fig:batch_n100}, for $n=100$, and
Figure~\ref{fig:batch_n10000} for $n=10,000$. Comparing both figures, one observes that OL06a is
the only approach which appears to be consistent with a variance compatible with the asymptotic
interquartile range shown on the right of each plot. EM5 (five iterations of batch EM) is
clearly the method which has the less variability but Figure~\ref{fig:batch_n10000} suggests that
it is not $1/\sqrt{n}$ consistent, which was indeed confirmed using longer runs not shown here.
This observation supports the claim of \cite{neal:hinton:1999} that, for
large sample sizes, online EM approaches are more efficient, from a computational point of view, than the batch EM algorithm which requires several iterations to converge properly.
The online EM with step size $\gamma_n = 1/n$ (OL1) presents a bias which becomes very significant
when $n$ increases. According to Theorem~\ref{theo:weak-convergence-REM}, this problem could be
avoided (asymptotically) by choosing a sufficiently small value of $\gamma_0$. For fixed $n$
however, lowering $\gamma_0$ can only reduce the perceived speed of convergence, which is already
very slow, as illustrated by Figure~\ref{fig:batch_traj}. In contrast, the online EM algorithm with
Polyak-Ruppert averaging (OL06a) appears to be very efficient: averaging significantly reduces the
variability of the OL06 estimate, reducing it to a level which is consistent with the asymptotic
interquartile range, while maintaining a systematic bias which vanishes as $n$ increases, as
expected.

%==================================================================
\section{Conclusion}
Compared to other alternatives, the main advantages of the proposed approach to online parameter
estimation in latent data models are its analogy with the standard batch EM algorithm, which makes
the online algorithm easy to implement, and its provably optimal convergence behaviour. In
addition, the combination of a slow parameter decrease ($\gamma_n = n^{-0.5+\epsilon}$ being a
typical choice) with Polyak-Ruppert averaging appears to be very robust.

A limitation is the fact that the function $\functheta{s}$ has to be explicit, which, for instance,
would not be the case for mixture of regression models with generalised link functions. Another
extension of interest concerns non independent models and in particular hidden Markov models or
Markov random fields.

\bibliography{RecursiveEM}
\bibliographystyle{abbrvnat}

%==================================================================
\appendix
\section{Proofs}
\label{sec:proofs}
For $\mu = (\mu_1, \dots, \mu_m)^T$ a differentiable function from $\Theta$ to $\rset^m$, we define
by $\nabla_\theta \mu^T$ the $d_\theta \times m$ matrix whose columns are the gradients,
$\nabla_\theta \mu^T = \left[\nabla_\theta \mu_1, \dots, \nabla_\theta \mu_m\right]$. Thus the
symbol $\nabla_\theta$ denotes either a vector or a matrix, depending on whether the function to
which it is applied is scalar or vector-valued. With this convention, the matrix $\nabla_\theta
\mu^T$ is the transpose of the usual \emph{Jacobian matrix}.

\begin{proof}[Proof of Proposition \ref{prop:characterization-stationary-points}]
  Let $s^\star$ be a root of $\meanfield$, and set
  $\theta^\star= \functheta{s^\star}$. For any $s \in \calS$, the function
  $\theta \mapsto \ls(s;\theta)$ has a unique stationary point at
  $\functheta{s}$. Therefore,
  \begin{equation}
    \label{eq:keyrelation1}
    -\nabla_\theta \psi(\theta^\star) + \nabla_\theta \phi^T(\theta^\star) s^\star=0 \eqsp.
  \end{equation}
  Note that, from \emph{Fisher's identity},
  \begin{equation}
    \label{eq:FisherIdentity}
    \dly(y;\theta)= \PE_\theta \left[ \nabla_\theta \log \fx(X;\theta) \right| \left. Y=y \right] = -\nabla_\theta \psi(\theta) + \nabla_\theta
    \phi^T(\theta) \condexpsuffstat{y}{\theta} \eqsp.
  \end{equation}
  As $\nabla_\theta \kullback{\pi}{\fy_\theta} = \PE_\pi \left[ \dly(Y;\theta) \right]$, the latter relation implies
  \begin{equation}
    \label{eq:keyrelation2}
    \nabla_\theta \kullback{\pi}{\fy_\theta}= - \nabla_\theta \psi(\theta) + \nabla_\theta \phi^T(\theta) \PE_\pi \left[ \condexpsuffstat{Y}{\theta} \right] \eqsp.
  \end{equation}
  Since $\meanfield(s^\star)= 0$, $s^\star = \PE_\pi[ \condexpsuffstat{Y}{\theta^\star} ]$,
  and thus \eqref{eq:keyrelation1} and \eqref{eq:keyrelation2} imply
  \[
  \left. \nabla_\theta \kullback{\pi}{\fy_\theta} \right|_{\theta=
    \theta^\star}= - \nabla_\theta \psi(\theta^\star) + \nabla_\theta
  \phi^T(\theta^\star) s^\star= 0 \eqsp,
  \]
  establishing the first assertion. Conversely, suppose that $ \left.
    \nabla_\theta \kullback{\pi}{\fy_\theta} \right|_{\theta= \theta^\star}$
  $=0$ and set $s^\star= \PE_\pi[ \condexpsuffstat{Y}{\theta^\star}]$. By \eqref{eq:keyrelation2},
$$- \nabla_\theta \psi(\theta^\star) + \nabla_\theta \phi^T(\theta^\star) s^\star= 0 \eqsp.$$
Under Assumption~\ref{assum:expon}(c), the function $\theta
\mapsto -\psi(\theta) + \pscal{\phi(\theta)}{s^\star}$ has a unique stationary
point at $\functheta{s^\star}$, which is a maximum. Hence, $\theta^\star =
\functheta{s^\star}$ establishing the second assertion.
\end{proof}

\begin{proof}[Proof of Proposition \ref{prop:lyapunov-function-h}]
  Using \eqref{eq:keyrelation2} and the chain rule of differentiation,
  \begin{equation}
    \label{eq:gradient-lyapunov}
    \nabla_s \lyapunov(s) = - \nabla_s \functheta[T]{s} % \\
    \, \left\{ - \nabla_\theta \psi(\functheta{s}) + \nabla_\theta
      \phi^T(\functheta{s}) \PE_\pi \left[\condexpsuffstat{Y}{\functheta{s}} \right] \right\} \eqsp.
  \end{equation}
  For any $s \in \calS$, $\functheta{s}$ is the maximum of $\theta \mapsto \ls(s;\theta)$, thus,
  \begin{equation}
    \label{eq:keyrelation3}
    \left. \nabla_\theta \ls(s;\theta) \right|_{\theta= \functheta{s}} = 0 = - \nabla_\theta \psi(\functheta{s}) +
    \nabla_\theta \phi^T(\functheta{s}) s \eqsp.
  \end{equation}
  Plugging this relation into \eqref{eq:gradient-lyapunov} and using the
  definition \eqref{eq:recursiveEM-curvedexponential-meanfield}, we obtain
  \begin{equation}
    \label{eq:gradient-lyapunov-h}
    \nabla_s \lyapunov(s) = - \nabla_s \functheta[T]{s}  \nabla_\theta \phi^T (\functheta{s}) \meanfield(s) \eqsp.
  \end{equation}
  By differentiating the function $s \mapsto \Phi(s;\functheta{s})$ where
  $\Phi(s,\theta)\eqdef \nabla_\theta \ls(s; \theta)$, we obtain
  \[
  \nabla_s \Phi^T(s; \functheta{s}) = \nabla_s \Phi^T(s;\functheta{s}) +  \nabla_s \functheta[T]{s} \nabla_\theta \Phi^T(s,\functheta{s}) \eqsp.
  \]
  Since $\nabla_s \Phi^T(s;\theta)= \nabla_s \left[\nabla_\theta \ls(s;\theta)\right]^T =
  \left[\nabla_\theta \phi^T(\theta)\right]^T$ and $\nabla_\theta \Phi^T(s,\theta)= \nabla_\theta^2 \ell(s,\theta)$,
  the latter relation may be alternatively written as
  \[
  \nabla_s \left[\left. \nabla_\theta \ls(s; \theta) \right|_{\theta= \functheta{s}} \right]^T
  = \left[ \nabla_\theta \phi^T(\functheta{s}) \right]^T +
  \nabla_s \functheta[T]{s} \; \left. \nabla^2_\theta \ls(s;\theta) \right|_{\theta= \functheta{s}}  \eqsp.
  \]
  Because the function $s \mapsto \left. \nabla_\theta \ls(s; \theta)
  \right|_{\theta= \functheta{s}}$ is identically equal to 0, $\nabla_s
  [ \left. \nabla_\theta \ls(s; \theta) \right|_{\theta= \functheta{s}} ]^T= 0$, and
  thus,
  \begin{equation}
    \label{eq:derivative-htheta}
    \nabla_\theta \phi^T(\functheta{s}) = -\left. \nabla^2_\theta \ls(s;\theta)  \right|_{\theta= \functheta{s}} \left[ \nabla_s \functheta[T]{s} \right]^T \eqsp.
  \end{equation}
  Plugging this relation into \eqref{eq:gradient-lyapunov-h} yields
  \begin{equation}
    \label{eq:keyrelation4}
    \pscal{\nabla_s \lyapunov(s)}{\meanfield(s)} = \meanfield^T(s) \left[\nabla_\theta \phi^T (\functheta{s})\right]^T \left\{ \left. \nabla^2_\theta \ls(s;\theta)  \right|_{\theta= \functheta{s}} \right\}^{-1} \nabla_\theta \phi^T (\functheta{s}) \meanfield(s) \eqsp,
  \end{equation}
  where assumption~\ref{assum:expon}(c) implies that, for any $s \in \calS$, the matrix $\left. \nabla^2_\theta \ls(s;\theta)
  \right|_{\theta= \functheta{s}}$ is negative definite (and hence invertible). Hence, for any $s
  \in \calS$, $\pscal{\nabla_s \lyapunov(s)}{\meanfield(s)} \leq 0$
  with equality if and only if $\nabla_\theta \phi^T (\functheta{s}) \meanfield(s) = 0$.

  To deal with the equality case, assume that $s^\star$ is such that $\nabla_\theta \phi^T (\functheta{s^\star}) \meanfield(s^\star) = 0$, or equivalently that
  \[
     \nabla_\theta \phi^T (\functheta{s^\star}) \PE_\pi \left[ \condexpsuffstat{Y}{\functheta{s^\star}} \right] = \nabla_\theta \phi^T (\functheta{s^\star}) s^\star \eqsp .
  \]
  Under Assumptions~\ref{assum:expon}(a) and~\ref{assum:expon}(c), $\theta^\star =
  \functheta{s^\star}$ is the unique solution to the score equation~\eqref{eq:keyrelation3},
  that is, such that $\nabla_\theta \psi(\theta^\star) = \nabla_\theta \phi^T (\theta^\star)
  s^\star$. Hence, $\nabla_\theta \phi^T (\theta^\star) \PE_\pi \left[
    \condexpsuffstat{Y}{\theta^\star} \right] = \nabla_\theta \psi(\theta^\star)$, which implies
  that $\functheta{\PE_\pi \left[ \condexpsuffstat{Y}{\theta^\star} \right]} = \theta^\star$
  showing that $\theta^\star \in \mathcal{L}$ and thus, from
  Proposition~\ref{prop:characterization-stationary-points}, that $s^\star \in \calS$. The proof
  then follows upon noting that $s \mapsto \pscal{\nabla_s \lyapunov(s)}{\meanfield(s)}$ is
  continuous.
\end{proof}

\begin{proof}[Proof of Theorem \ref{theo:wp1convergencerecursiveEM}]
Under the stated assumptions, for any $\epsilon > 0$, there exists a compact $\calK \subset \calS$ and $n$, such that, 
$\PP \left( \bigcap_{k \geq n} \{ \hcondexpsuffstat_n \in \calK \} \right) \geq 1-\epsilon$.   Therefore, for any $\eta > 0$,
\begin{align*}
& \PP\left( \sup_{k \geq n} \left| \sum_{i=n}^k \gamma_i \xi_i \right| \geq \eta\right)   \\
& \qquad  \leq \PP \left( \sup_{k \geq n} \left| \sum_{i=n}^k \gamma_i \xi_i \1_{\calK}(\hcondexpsuffstat_i) \right| \geq \eta, \bigcap_{i \geq n} \{\hcondexpsuffstat_i \in \calK \} \right) 
+ \PP \left( \bigcup_{i \geq n} \{ \hcondexpsuffstat_i \notin  \calK \} \right) \eqsp, \\
&\qquad \leq \epsilon + \PP \left( \sup_{k \geq n} \left| \sum_{i=n}^k \gamma_i \xi_i \1_{\calK}(\hcondexpsuffstat_i) \right| \geq \eta \right) \eqsp.
\end{align*}
Note that $M_{n,k} = \sum_{i=n}^k \gamma_i \xi_i \1_{\calK}(\hcondexpsuffstat_i)$ is a
$L_2$-martingale, and that its angle-bracket is bounded by $ \sup_{s \in \calK} \PE_\pi[
|\condexpsuffstat{Y}{\functheta{s}}|^2] \sum_{i=n}^k \gamma_i^2 < \infty$.  Using Chebyshev's
inequality associated to the Doob martingale inequality, we conclude that
$$\PP\left( \sup_{k \geq n} |M_{n,k}|\geq \eta\right)  \leq 2 \eta^{-2} \sup_{s \in \calK} \PE_\pi[ |\condexpsuffstat{Y}{\functheta{s}}|^2] \sum_{i=n}^\infty \gamma_i^2 \eqsp,
$$
finally showing that $\lim\sup_n \sup_{k \geq n} |M_{n,k}| = 0$ with probability one. The proof is concluded by applying Theorem 2.3 of \cite{andrieu:moulines:priouret:2005} which states that the sequence of sufficient statistics defined by~\eqref{eq:recursiveEM-curvedexponential-RM} is then such that $\lim\sup_n d(\hcondexpsuffstat_n,\Gamma)= 0$ with probability one; the result on the sequence of parameter estimates $\htheta_n$ follows by continuity of $\bar{\theta}$.
\end{proof}

To prove Proposition \ref{prop:asymptoticequivalencerecursiveEMs}, we will make use of the following stability lemma.

\begin{lemma}
\label{lem:bounded-in-probability}
Let $p \geq 1$. Assume that for any compact subset $\mathcal{K} \subset \calS$, $\sup_{s \in \calK} \PE_\pi \left[ |\condexpsuffstat{Y}{\functheta{s}}|^p \right] < \infty$ for some $p>0$ and that $\PP_\pi$-\as\ $\lim \hcondexpsuffstat_n$ exists and belongs to $\calS$. Then, the sequence $\{ \condexpsuffstat{Y_{n+1}}{\functheta{\hcondexpsuffstat_n}} \}_{n \geq 0}$ is bounded in probability, \ie\
\[
\lim_{M \to \infty} \limsup_{n \to \infty} \PP\left( |\condexpsuffstat{Y_{n+1}}{\functheta{\hcondexpsuffstat_n}|}| \geq M \right) = 0 \eqsp.
\]
\end{lemma}
\begin{proof}
Let $\calK$ be a compact subset of $\calS$. We may decompose $\PP\left( |\condexpsuffstat{Y_{n+1}}{\functheta{\hcondexpsuffstat_n}|}| \geq M \right)$
as follows
\begin{align*}
\PP\left( |\condexpsuffstat{Y_{n+1}}{\functheta{\hcondexpsuffstat_n}|}| \geq M \right)
&\leq \PP( \hcondexpsuffstat_n \not \in \calK) +  \PP \left( |\condexpsuffstat{Y_{n+1}}{\functheta{\hcondexpsuffstat_n}|}| \geq M , \hcondexpsuffstat_n \in \calK \right) \eqsp , \\
&\leq \PP( \hcondexpsuffstat_n \not \in \calK) + M^{-p} \sup_{s \in \calK} \PE_\pi \left[ |S(Y, \functheta{s})|^p \right] \PP\left( \hcondexpsuffstat_n \in \calK \right) \eqsp,
\end{align*}
which implies that
\[
\limsup_{n \to \infty} \PP\left( |\condexpsuffstat{Y_{n+1}}{\functheta{\hcondexpsuffstat_n}|}| \geq M \right) \leq
\PP(\lim_{n \to \infty} \hcondexpsuffstat_n \not \in \calK) + M^{-p} \sup_{s \in \calK} \PE_\pi \left[ |S(Y, \functheta{s})|^p \right] \eqsp.
\]
\end{proof}

\begin{proof}[Proof of Proposition \ref{prop:asymptoticequivalencerecursiveEMs}]
  A Taylor expansion with integral remainder shows that
  \begin{align}
    \label{eq:TaylorExpansion} \nonumber
    \htheta_{n+1}&= \functheta{ \hcondexpsuffstat_n + \gamma_{n+1} \left( \condexpsuffstat{Y_{n+1}}{\htheta_n} - \hcondexpsuffstat_n \right) }\\
    &= \htheta_n + \gamma_{n+1} \left[\nabla_s \functheta[T]{\hcondexpsuffstat_n}\right]^T \left(
      \condexpsuffstat{Y_{n+1}}{\htheta_n} - \hcondexpsuffstat_n \right) +
       \gamma_{n+1} r_{n+1} \eqsp,
  \end{align}
  where the remainder $r_{n+1}$ is given by
  \begin{multline}
    \label{eq:RemainderTaylorExpansion}
    r^T_{n+1} \eqdef \left( \condexpsuffstat{Y_{n+1}}{\htheta_n} - \hcondexpsuffstat_n \right)^T \\
     \times \int_0^1 \left[\nabla_s \functheta[T]{\hcondexpsuffstat_{n} + t \gamma_{n+1} \left\{ \condexpsuffstat{Y_{n+1}}{\htheta_n} - \hcondexpsuffstat_n \right\} } -
    \nabla_s \functheta[T]{\hcondexpsuffstat_{n}} \right]
      \rmd t \eqsp .
  \end{multline}

  We will first show that $\lim_{n \to \infty} r_{n} = 0$ \as\ 
  Lemma \ref{lem:bounded-in-probability} shows that $\{ \condexpsuffstat{Y_{n+1}}{\htheta_n} \}_{n \geq
    0}$ is bounded in probability, which we denote by $\condexpsuffstat{Y_{n+1}}{\htheta_n} = O_\PP(1)$.  Under the assumption of
  Theorem \ref{theo:wp1convergencerecursiveEM}, $\hcondexpsuffstat_n =
  O_\PP(1)$, which implies that $\condexpsuffstat{Y_{n+1}}{\htheta_n} -
  \hcondexpsuffstat_n = O_\PP(1)$.
  Choose $\epsilon > 0$ and then a compact subset $\calK$ and a constant $M$ large enough such that
  \begin{equation}
  \label{eq:definition-compact-and-M}
  \limsup_{n \to \infty} \PP\left( \hcondexpsuffstat_n \not \in \calK \right) + \limsup_{n \to \infty} \PP\left(\left| \condexpsuffstat{Y_{n+1}}{\htheta_n} - \hcondexpsuffstat_n \right| \geq M\right) \leq \epsilon \eqsp .
  \end{equation}
  Since the function $\nabla_s \functheta{\cdot}$ is assumed continuous, it is uniformly continuous over every compact subset, \ie, there exists a constant $\delta_0$ such that
  $\calK_{\delta_0} \eqdef \left\{ s \in \calS, d(s,\calK) \leq \delta_0 \right\} \subset \calS$ and
  \begin{equation} \label{eq:localuniformcontinuityfunctheta}
  \sup_{|h| \leq \delta_0, s \in \calK} |\nabla_s \functheta{s+h} - \nabla_s
  \functheta{s}| \leq \epsilon \eqsp.
  \end{equation}
Since $\lim_{n \to \infty} \gamma_n = 0$ and
  $( \condexpsuffstat{Y_{n+1}}{\htheta_n} - \hcondexpsuffstat_n )$
  is bounded in probability, $\gamma_{n+1} (
    \condexpsuffstat{Y_{n+1}}{\htheta_n} - \hcondexpsuffstat_n )$
  converges to zero in probability, which we denote by $\gamma_{n+1} (
    \condexpsuffstat{Y_{n+1}}{\htheta_n} - \hcondexpsuffstat_n ) =
  o_\PP(1)$. For $\delta_0>0$ satisfying (\ref{eq:localuniformcontinuityfunctheta}), this implies in particular that
  \[
    \lim_{n \to \infty} \PP\left( \gamma_{n+1}
    \left| \condexpsuffstat{Y_{n+1}}{\htheta_n} - \hcondexpsuffstat_n \right|
    \geq \delta_0\right)= 0 \eqsp .
  \]
  Therefore,
  \begin{multline*}
    \limsup_{n \to \infty} \PP\left(| r_{n+1}| \geq M \epsilon\right) \leq
    \limsup_{n \to \infty} \PP\left( \gamma_{n+1} \left|
        \condexpsuffstat{Y_{n+1}}{\htheta_n} - \hcondexpsuffstat_n \right| \geq
      \delta_0 \right) \\ + \limsup_{n \to \infty} \PP\left( \hcondexpsuffstat_n \not \in \calK \right) +
      \limsup_{n \to \infty} \PP \left(\left| \condexpsuffstat{Y_{n+1}}{\htheta_n} - \hcondexpsuffstat_n \right| \geq M\right) \leq \epsilon \eqsp,
  \end{multline*}
  showing that $r_n = o_\PP(1)$.

We now proceed with the first order term in \eqref{eq:TaylorExpansion}. From \eqref{eq:derivative-htheta} we have,
  \begin{equation}
    \label{eq:derivative-htheta-2}
    \nabla_s \functheta[T]{s} = \left[ \nabla_\theta \phi^T(\functheta{s}) \right]^T \left[ - \left. \nabla^2_\theta \ls(s ; \theta) \right|_{\theta=\functheta{s}} \right]^{-1}   \eqsp.
  \end{equation}
In addition, \eqref{eq:FisherIdentity} and
  \eqref{eq:keyrelation3} show that
  \begin{equation}
    \label{eq:FisherIdentity-2}
    \nabla_\theta \phi^T \left(  \functheta{s} \right) \left\{ \condexpsuffstat{y}{\functheta{s}} -s \right\} = \dly(y ; \functheta{s}) \eqsp.
  \end{equation} 
  Therefore, by combining \eqref{eq:derivative-htheta-2} and \eqref{eq:FisherIdentity-2}, the first order term in \eqref{eq:TaylorExpansion} may be rewritten as
  \begin{multline}
    \left[ \nabla_s \functheta[T]{\hcondexpsuffstat_n} \right]^T\left(
      \condexpsuffstat{Y_{n+1}}{\functheta{\hcondexpsuffstat_n}} -
      \hcondexpsuffstat_n \right) = \FIMcomplete[-1]{\pi}{\htheta_n}
    \dly(Y_{n+1} ; \htheta_n) \\ + \left \{ \left[ - \left. \nabla^2_\theta
          \ls(\hcondexpsuffstat_n ; \theta) \right|_{\theta=\htheta_n}
      \right]^{-1} - \FIMcomplete[-1]{\pi}{\htheta_n} \right\} \dly(Y_{n+1} ;
    \htheta_n) \eqsp.
    \label{eq:tmp:first-order-term}
  \end{multline}
  By definition, the complete data FIM may be rewritten as $\FIMcomplete{\pi}{\theta} = \left. \nabla^2_\theta \ls \left( s; \theta \right) \right|_{s =
    \PE_\pi[\condexpsuffstat{Y}{\theta}]}$. Note also that $\htheta_n$ converges a.s. to $\theta^\star$ and that $\FIMcomplete{\pi}{\theta^\star}$ is assumed to be positive definite. Hence, to ensure that the term between braces in~\eqref{eq:tmp:first-order-term} may be neglected, we need to show that:
  \begin{equation}
    \label{eq:negligeability}
    \left. \nabla^2_\theta \ls( \hcondexpsuffstat_n; \theta) \right|_{\theta= \htheta_n} -  \left. \nabla^2_\theta \ls( s; \theta) \right|_{(s,\theta)=
      (\PE[\condexpsuffstat{Y_{n+1}}{\htheta_n}|\mcf_n],\htheta_n)} = o_\PP(1) \eqsp.
  \end{equation}
   Since the function $(s,\theta) \mapsto \nabla^2_\theta \ls(s ; \theta)$ is continuous,
   there exists $\delta_1 > 0$ such that, $\calK_{\delta_1}= \{ s \in \calS, d(s,\calK) \leq \delta_1 \} \subset \calS$
   and
  \[
  \sup_{|h| \leq \delta_1,  s \in \calK} \left| \nabla_\theta^2 \ls( s+h; \functheta{s}) - \nabla_\theta^2 \ls(s;\functheta{s}) \right| \leq \epsilon \eqsp,
  \]
  where the set $\calK$ is defined in \eqref{eq:definition-compact-and-M}. Under the stated assumption
  $\lim_{n \to \infty} d(\hcondexpsuffstat_n,\mathcal{L})= 0$ \as, which implies
  that
  \[
  \lim_{n \to \infty} \meanfield( \hcondexpsuffstat_n)= \lim_{n \to \infty} \{
  \PE[\condexpsuffstat{Y_{n+1}}{\functheta{\hcondexpsuffstat_n}}| \mcf_{n}] - \hcondexpsuffstat_n \}= 0 \eqsp, \quad
  \text{\as\ }
  \]
  By combining the latter two equations, we therefore obtain
  \begin{multline*}
    \limsup_{n \to \infty}  \PP \left( \left|\left. \nabla^2_\theta \ls( \hcondexpsuffstat_n; \theta)
      \right|_{\theta= \htheta_n} - \left. \nabla^2_\theta \ls( s; \theta)
      \right|_{(s,\theta)=
        (\PE[\condexpsuffstat{Y_{n+1}}{\htheta_n}|\mcf_n],\htheta_n)} \right| \geq \epsilon \right) \leq \\
    \limsup_{n \to \infty} \PP \left( \hcondexpsuffstat_n \not \in \calK \right) + \limsup_{n \to \infty} \PP
    \left( | \meanfield( \hcondexpsuffstat_n) | \geq \delta_0 \right) \leq \epsilon \eqsp,
  \end{multline*}
  showing \eqref{eq:negligeability}.
\end{proof}

\begin{proof}[Proof of Theorem \ref{theo:weak-convergence-REM}]
  We use the definition of the recursive
  EM sequence given by Proposition
  \ref{prop:asymptoticequivalencerecursiveEMs}. The mean field associated to
  this sequence is given by $\meanfield_\theta(\theta) \eqdef -
  \FIMcomplete[-1]{\pi}{\theta} \nabla_\theta \kullback{\pi}{\fy_\theta}$.  The
  Jacobian of this vector field at $\theta^\star$ is equal to
  $H(\theta^\star)$.  Since $\theta^\star$ is a (possibly local) minimum of the
  Kullback-Leibler divergence, the matrix $ \left. \nabla^2_\theta
    \kullback{\pi}{\fy_\theta} \right|_{\theta= \theta^\star}$ is positive
  definite.  Because the two matrices $\FIMcomplete[-1]{\pi}{\theta^\star}
  \left. \nabla^2_\theta \kullback{\pi}{\fy_\theta} \right|_{\theta=\theta^\star}$ and
  $\FIMcomplete[-1/2]{\pi}{\theta^\star} \left. \nabla^2_\theta
  \kullback{\pi}{\fy_\theta} \right|_{\theta=\theta^\star}
  \FIMcomplete[-1/2]{\pi}{\theta^\star}$
  have the same eigenvalues, counting
  multiplicities, the eigenvalues of the matrix
  $H(\theta^\star) = \FIMcomplete[-1]{\pi}{\theta^\star} \left. \nabla^2_\theta
  \kullback{\pi}{\fy_\theta} \right|_{\theta=\theta^\star}$
  are all real and strictly positive. This shows
  the first assertion of the theorem; the proof of
  the second assertion follows directly from \citealp[Theorem 1]{pelletier:1998}.
\end{proof}

\section*{Acknowledgement}
The authors thank Maurice Charbit and Christophe Andrieu for
precious feedback during the early stages of this work.

\end{document}